\documentclass[letterpaper, 12pt, journal, onecolumn]{IEEEtran}

\usepackage{cite}
\usepackage{amsmath,amssymb,amsfonts}
\usepackage{graphicx}
\usepackage{textcomp}
\usepackage{enumerate}
\usepackage[usenames,dvipsnames,table]{xcolor}
\usepackage{soul}
\usepackage[normalem]{ulem}
\usepackage{url}
\usepackage{graphicx}
\usepackage{color}
\usepackage{cite}
\usepackage{amsfonts}
\usepackage{amsthm}
\usepackage{algorithm}
\usepackage{mathtools}
\usepackage{tikz}
\usepackage{pgfplots}
\usepgfplotslibrary{groupplots}
\usepackage{bm}
\usepackage{algorithm}
\usepackage{algpseudocode}
\usepackage{tcolorbox}
\usetikzlibrary{shapes.symbols}
\usetikzlibrary{3d, arrows.meta}
\usepackage{tikz-3dplot}
\usepackage{diagbox}

\newcommand\oprocendsymbol{\hbox{$\square$}}
\newcommand\oprocend{\relax\ifmmode\else\unskip\hfill\fi\oprocendsymbol}
\def\eqoprocend{\tag*{$\square$}}

\DeclareMathOperator{\dom}{dom}

\newcommand{\rank}{\operatornamewithlimits{rank}}

\newcommand{\im}{\operatornamewithlimits{im}}

\newtheorem{theorem}{Theorem}

\newtheorem{lemma}{Lemma}
\newtheorem{remark}{Remark}
\newtheorem{definition}{Definition}

\newcommand{\R}{{\mathbb{R}}} %
\newcommand{\N}{{\mathbb{N}}} %
\newcommand{\hw}{{\hat{w}}} %
\newcommand{\tu}{{\tilde{u}}} %
\newcommand{\tx}{{\tilde{x}}} %
\newcommand{\tw}{{\tilde{w}}} %
\newcommand{\hctrbx}{{\mathbb{L}_{\textup{x}}^{\textup{\tiny C}}}} %
\newcommand{\hctrbxsi}{{\mathbb{L}_{\textup{x, 1}}^{\textup{\tiny C}}}} %
\newcommand{\hctrbxmi}{{\mathbb{L}_{\textup{x}, \textup{m}}^{\textup{\tiny C}}}} %
\newcommand{\hctrbxu}{{\mathbb{L}_{\textup{xu}}^{\textup{\tiny C}}}} %
\newcommand{\hctrbxusi}{{\mathbb{L}_{\textup{xu, 1}}^{\textup{\tiny C}}}} %
\newcommand{\hctrbxumi}{{\mathbb{L}_{\textup{xu}, \textup{m}}^{\textup{\tiny C}}}} %

\newcommand{\sctrbx}{{\mathbb{L}_{\textup{x}}^{\textup{\tiny D}}}} %
\newcommand{\sctrbxsi}{{\mathbb{L}_{\textup{x, 1}}^{\textup{\tiny D}}}} %
\newcommand{\sctrbxmi}{{\mathbb{L}_{\textup{x}, \textup{m}}^{\textup{\tiny D}}}} %
\newcommand{\sctrbxu}{{\mathbb{L}_{\textup{xu}}^{\textup{\tiny D}}}} %
\newcommand{\sctrbxusi}{{\mathbb{L}_{\textup{xu, 1}}^{\textup{\tiny D}}}} %
\newcommand{\sctrbxumi}{{\mathbb{L}_{\textup{xu}, \textup{m}}^{\textup{\tiny D}}}} %

\newcommand{\dw}{{\dot{w}}} %
\newcommand{\dx}{{\dot{x}}} %
\newcommand{\du}{{\dot{u}}} %

\newcommand{\dimx}{{n}} 
\newcommand{\dimu}{{m}}

\newcommand{\PE}{WPE} %

\newcommand{\C}[1]{{\mathcal{C}^\infty_b(\R^{#1 })}}
\newcommand{\D}[1]{{\ell_\infty(\R^{#1})}}
\newcommand{\Cpe}[1]{\Omega^{\textup{\tiny C}}_{#1}}
\newcommand{\Dpe}[1]{\Omega^{\textup{\tiny D}}_{#1}}

\newcommand{\Cppe}[2]{\Omega^{\textup{\tiny C}}_{#1, #2}}
\newcommand{\Dppe}[2]{\Omega^{\textup{\tiny D}}_{#1, #2}}

\newcommand{\Csr}{{\Psi^{\textup{\tiny C}}}}
\newcommand{\Dsr}{{\Psi^{\textup{\tiny D}}}}

\newcommand{\xx}{{{\bm{x}}}} %
\newcommand{\UU}{{{\bm{U}}}} %
\newcommand{\WW}{{{\bm{W}}}} %
\newcommand{\yy}{{{\bm{y}}}} %
\newcommand{\uu}{{{\bm{u}}}} %
\newcommand{\vv}{{{\bm{v}}}} %
\newcommand{\ww}{{{\bm{w}}}} %

\newcommand{\bw}{{\bar{w}}} %
\newcommand{\bx}{{\bar{x}}} %

\newcommand{\qq}{{\mathrm{q}}} %
\newcommand{\QQ}{{\mathrm{Q}}} %
\newcommand{\dd}{{\mathrm{d}}} %
\newcommand{\DD}{{\mathrm{D}}} %

\newcommand{\tet}{t} %
\newcommand{\tU}{\tilde{U}} %

\title{\Large On Sufficient Richness for 
	Linear Time-Invariant Systems}

\begin{document}

\author{Marco Borghesi, Simone Baroncini, Guido Carnevale, Alessandro Bosso, and Giuseppe Notarstefano
	\thanks{
		Funded by the European Union - PRIN 2022 ECODREAM Energy COmmunity management: DistRibutEd AlgorithMs and toolboxes for efficient and sustainable operations code [202228CTKY\_002] - CUP [J53D23000560006].
		M. Borghesi, S. Baroncini, G. Carnevale, A. Bosso, and G. Notarstefano are with the Department of Electrical, Electronic, and Information Engineering, Alma Mater Studiorum - Universit\`a di Bologna, 40136 Bologna, Italy. 
		The corresponding author is M.~Borghesi (e-mail: m.borghesi@unibo.it).}
}

\maketitle

\begin{abstract}
 	Persistent excitation (PE) is a necessary and sufficient condition
	for uniform exponential parameter convergence in several adaptive,
	identification, and learning schemes. In this article, we consider,
	in the context of multi-input linear time-invariant (LTI) systems,
	the problem of guaranteeing PE of commonly-used regressors by
	applying a sufficiently rich (SR) input signal. Exploiting the
	analogies between time shifts and time derivatives, we state simple
	necessary and sufficient PE conditions for the discrete- and
	continuous-time frameworks.  Moreover, we characterize the shape of
	the set of SR input signals for both single-input and multi-input
	systems.  Finally, we show with a numerical example that the derived
	conditions are tight and cannot be improved without including
	additional knowledge of the considered LTI system.
\end{abstract}

\begin{IEEEkeywords}
	\noindent
	Persistent excitation, sufficient richness, adaptive control, data-driven control, linear systems.
\end{IEEEkeywords}

\section{Introduction}\label{sec:introduction}

	\IEEEPARstart{70}{'s} were the years in which persistent excitation
	(PE) was recognized as a fundamental property in adaptive control and
	system identification \cite{morgan1977stability, morgan1977uniform,
		anderson1977exponential, narendra2012stable, sastry2011adaptive,
		ioannou1996robust}. While the original results were given in the
	continuous-time domain, they were soon extended to the discrete time
	\cite{bitmead1980lyapunov} (see \cite{narendra1987persistent} for a
	comprehensive review). In the early $2000$s, the conditions for PE
	were revisited and refined to include a larger class of dynamical
	systems and adaptive control schemes \cite{loria1999new,
		loria2000ugas, panteley2001relaxed, loria2002uniform,
		loria2003persistency, loria2005nested, saoud2024hybrid}.  These
	results guarantee that PE of certain internal signals (``regressors")
	implies uniform robust parameter convergence and thus are important
	for system identification \cite{aastrom1965numerical,
		aastrom1971system, ljung1990adaptation}, filtering
	\cite{sondhi1976new, weiss1979digital, goodwin1980discrete,
		anderson1982exponential, elliott1985global}, adaptive observers
	\cite{kreisselmeier1977adaptive, ioannou1996robust}, model reference
	adaptive control \cite{anderson1977approach, boyd1983parameter,
		boyd1986necessary}, gradient-based methods
	\cite{kudva1974identification, bitmead1984persistence}, and
	reinforcement learning \cite{possieri2022value, carnevale2023onpolicy,
		borghesi2024mr}.
	
	\noindent
	In the aforementioned literature, a signal is considered persistently
	exciting if, within a moving time window of fixed length, the signal
	spans all directions with a certain minimum amount of energy,
	uniformly over time.  A similar (though, different) notion of PE has
	been made famous by Willems et al. \cite{willems2005note} in $2005$.
	According to \cite{willems2005note}, PE is a property involving a
	finite amount of data, namely, a PE signal spans all directions within
	a given finite-time window.  To clarify the exposition, we denote this
	notion of PE as ``\PE{}" (``Willems' PE"). The interpretation of
	the \PE{} condition is different from the ``classical" PE, since \PE{}
	is seen as the ability of gathered data to represent all the
	trajectories of a linear time-invariant (LTI) system.
	\PE{} has been proved to be important for the effectiveness of
	recently developed data-driven controllers and reinforcement learning
	algorithms \cite{de2019formulas, van2020data, berberich2020data,
		van2020willems, de2023learning, lopez2023efficient}.
	Motivated by the above discussion, in this article we consider LTI
	systems and we derive conditions for an input signal to obtain a PE
	output. In accordance with part of the literature, we use the term
	sufficient richness (SR) to denote, for a dynamical system, the
	property of an input signal that guarantees PE of an output. In the
	following, we summarize the existing results pertaining to sufficient
	richness for both PE and \PE{}.
	
	\subsubsection{PE}
		In the context of single-input discrete-time systems, in
		\cite{aastrom1965numerical} a certain property on time shifts of an
		input signal (``PE of order") was proven to guarantee parameter
		convergence in maximum likelihood estimation algorithms. In
		\cite{ljung1971characterization}, the ``PE of order" concept was
		traduced into a condition on the spectrum of the input signal, later
		named as ``SR of order". In \cite{yuan1977probing} the authors
		consider a multi-input system with almost periodic and periodic input
		signal, and derive conditions stating how many sinusoids should be
		placed into the components of the input to obtain PE of the
		state-input pair. A more recent, similar result can be found in
		\cite{ioannou1996robust}. In \cite{anderson1982exponential}, a
		stochastic version of the ``PE of order" condition was found to prove
		convergence in single-input systems. In \cite{boyd1983parameter} and
		later in \cite{boyd1986necessary}, conditions on the spectral lines
		(and later spectrum) of the scalar input signal were found to achieve
		PE of single-input systems.  These spectral conditions were extended
		to the multivariable case in \cite{nordstrom1987persistency}.  In
		\cite{moore1983persistence} (and later in \cite{bai1985persistency,
			green1986persistence}), a sufficient condition on the time shifts of
		the input was found also for the multi-input case.
		In \cite{narendra1987persistent} it was shown that if a signal is SR
		for a single-input stable system, then it is SR for all single-input
		stable systems. In \cite{mareels1988persistency}, LTV continuous-time
		systems are approached, and the authors propose for the first time a
		definition for ``SR of order" involving a low-pass filtering of the
		input signal.  Lastly, we cite \cite{padoan2017geometric}, where the
		authors study when an autonomous system produces PE solutions
		depending on the initial conditions and the properties of the system.

	\subsubsection{\PE}
		For this notion, thanks to the finite amount of data, a matrix
		representation of signals, the Hankel matrix, is considered. In the
		corollaries of Willems' Lemma \cite{willems2005note}, a finite-time
		equivalent of the results in \cite{moore1983persistence,
			bai1985persistency} for PE is provided. Later, in
		\cite{markovsky2023persistency} it was proved for single-input systems
		that these conditions (on input time-shifts) are both necessary and
		sufficient to achieve \PE. In \cite{coulson2022quantitative} and
		\cite{berberich2023quantitative}, the authors robustify the definition
		of \PE{} by introducing a lower bound, and they show how to modify
		existing results to guarantee it through the input. At last, both in
		\cite{rapisarda2022persistency} and in \cite{lopez2022continuous}, a
		``\PE{} of order" notion for continuous-time systems is proposed (the
		first involving time derivatives and the second a ``sampled" Hankel
		matrix), and results on the input derivatives or on piecewise constant
		input signals to obtain \PE{} of state-input signals are given.

	\subsubsection*{Contribution}
	
		The main article contribution is twofold: \textbf{(i)} we find
		necessary and sufficient conditions on the input signal for obtaining
		persistently excited state (state-input) signals in LTI systems;
		\textbf{(ii)} we unify the existing literature by proposing a concept
		of sufficient richness clearly distinct from the one of persistency of
		excitation and by developing all the new results and proofs within a
		common notation for the discrete-time and continuous-time frameworks.
		
		\noindent
		We split the first contribution into three parts.  First, we derive
		new necessary conditions for PE in multi-input systems. To the
		authors' knowledge, the only necessary result in the literature is
		given for \PE{} in discrete-time single-input systems
		\cite[Thm. 3]{markovsky2023persistency}. In order to find these
		conditions, we leverage the concept of partial persistence of
		excitation -namely, signals that persistently span only subspaces of
		the whole space they live in-, which first appeared in
		\cite[Def. $3$]{narendra1987persistent} (under a different name) but,
		to the authors' knowledge, was not used for this purpose.  Second, we
		derive sufficient conditions to obtain PE in LTI systems. While the
		given proof is new, the result is partially known, in the sense that
		the discrete-time part of the result was given in
		\cite[Thm. $1$]{bai1985persistency} and the continuous-time part was
		given in \cite[Prop. $1$]{rapisarda2022persistency} only for the
		finite-time \PE{} notion. In particular, our result is different from
		the one in \cite{rapisarda2022persistency} since we ensure lower
		bounds on the PE condition uniformly in time along infinitely long
		system trajectories.  At last, we combine the obtained necessary and
		sufficient results to derive and discuss the shape of the set of
		sufficiently rich signals for all stable, reachable LTI systems. It
		turns out that, for single-input systems, this set has an explicit
		characterization (in accordance with
		\cite[Thm. $1$]{narendra1987persistent}).  We demonstrate the
		tightness of the theoretical results with numerical examples.

		\noindent
		The second contribution of this work has the following
		implications. First, by separating the concepts of PE and SR, we do
		not need to talk about “PE of order" or “SR of order". This
		distinction can help to merge slightly different definitions into a
		common terminology. Second, by unifying the results and the notation
		for the discrete-time and the continuous-time frameworks, we better
		highlight the structural properties which are intrinsic to linear
		systems. We believe this may be a fundamental step toward the
		extension of these results to nonlinear systems.

		The paper is organized as follows. 
		In Section~\ref{sec:problem_setup}, we introduce some preliminary definitions for the terminology which is used to present the results. Then, we formulate the concepts of PE and SR and we present the problem addressed in this paper.
		In Section~\ref{sec:main_result}, we give the main contribution of the article, namely, we state necessary and sufficient conditions for PE in linear systems. 
		Finally, Section~\ref{sec:example} corroborates our theoretical results with a numerical example. 
		All the proofs are in the Appendix.

	\subsubsection*{Notation} 
		the set of real numbers is denoted by
		$\R$, while $\R_{\geq 0}\coloneqq [0, +\infty)$. $\N$ is the set of natural numbers.
		$(\cdot)^\top$ denotes the transpose of a real vector (matrix). 
		A matrix $M\in \R^{n\times n}$ is positive definite in $\R^n$, denoted $M>0$, if $x^\top M x >0$ for all $x \in \R^n$, $x \neq 0$. 
		We use bold letters $\xx, \uu, \ww$ to denote time signals as elements of infinite-dimensional vector spaces, while we use $x, u, w$ to indicate elements of finite-dimensional vector spaces. $\D{n}$ is the space of all bounded functions $\ww:\N\to \R^n$, and $\C{n}$ is the space of all smooth functions $\ww:\R_{\geq 0}\to \R^n$ with bounded derivatives. We equip these sets with the $\mathcal{L}_\infty$ norm $\|\ww\|_\infty\coloneqq \sup_{t \in \N}|w_t|$ (resp., $\|\ww\|_\infty\coloneqq \sup_{t \in \R_{\geq0}}|w(t)|$ for the continuous time). We use the notation $(x, u)\in \R^{n+m}$ to denote a column vector stacking $x\in \R^n$ and $u\in \R^m$. Respectively, $(\xx, \uu)$ denotes the signal collecting $(x_t, u_t)$ for each time instant.

\section{Preliminaries and Problem Setup}\label{sec:problem_setup}

	\subsection{Preliminaries}
	
		We introduce the main definitions of the paper. We start from the
		notion of bounded-input bounded-output dynamical systems in both
		discrete and continuous time.
		\begin{definition}[BIBO dynamical system - DT] 
			A bounded-input bounded-output (BIBO) discrete-time (DT) dynamical system is a map
			\begin{equation}
				\begin{split}
					\Sigma: \D{m} \times \R^n &\to \D{p} \\
					\uu, x_0 &\mapsto \yy
				\end{split}
			\end{equation}
			that can be written as
			\begin{equation}
				\begin{split}
					x_{t+1} &= f(x_t, u_t), \qquad x_0 \in \R^n,\\
					y_t &= h(x_t, u_t),
				\end{split}
			\end{equation}
			where $u_t\in \R^m, x_t\in \R^n, y_t\in \R^p$ are the input, the state, and the output at time $t\in \N$, while $f, h$ are the dynamic and the output map. 
			\oprocend
		\end{definition}

		\begin{definition}[BIBO dynamical system - CT] 
			A bounded-input bounded-output (BIBO) continuous-time (CT) dynamical system is a map
			\begin{equation}
				\begin{split}
					\Sigma: \C{m} \times \R^n &\to \C{p}\\
					\uu, x(0) &\mapsto \yy
				\end{split}
			\end{equation}
			that can be written as
			\begin{equation}
				\begin{split}
					\dx(t) &= f(x(t), u(t)), \qquad x(0) \in \R^n,\\
					y(t) &= h(x(t), u(t)),
				\end{split}
			\end{equation}
			where $u(t)\in \R^m, x(t)\in \R^n, y(t)\in \R^p$ are the input, the state, and the output at time $t\in \R_{\geq 0}$, while $f, h$ are the dynamic and the output map. 
			\oprocend
		\end{definition}
		
		\noindent
		Next, we introduce the shift and derivative operators, which will be used to state the main results of the article. 
		
		\begin{definition}[Shift operators]\label{def:shift_operator}
			Given $k, d \in \N$, we define the shift operator as
			\begin{equation}
				\begin{split}
					\qq^k: \D{d} &\to \D{d}\\
					w_t & \longmapsto 
					\begin{cases}
						w_{t-k} & t\geq k\\
						0 & t< k,
					\end{cases} \quad\forall t \in \N.\label{eq:shift}
				\end{split}
			\end{equation}	
			Furthermore, we introduce the following multi-shift operator:
			\begin{equation}
				\begin{split}
					\QQ^k: \D{d} &\to \D{dk}\\
					\ww & \longmapsto (\qq^{k-1}\ww, \ldots, \qq^0 \ww).
				\end{split}
			\end{equation}
			\oprocend
		\end{definition}

		\begin{definition}[Derivative operators]
			Given $k, d \in \N$, we define the derivative operator as
			\begin{equation}
				\begin{split}
					\dd^k: \C{d} &\to \C{d}\\
					w(t) & \longmapsto \frac{d^k w}{dt^k}(t), \quad \forall t \in \R_{\geq 0}.
				\end{split}
			\end{equation}	
			Furthermore, we introduce the following multi-derivative operator:
			\begin{equation}
				\begin{split}
					\DD^k: \C{d} &\to \C{dk}\\
					\ww & \longmapsto (\dd^{k-1}\ww, \ldots, \dd^0 \ww).
				\end{split}
			\end{equation}	
			\oprocend
		\end{definition}

	\subsection{Persistency of Excitation and Sufficient Richness}

		We begin by providing the formal definition of a PE signal in both discrete and continuous-time.
		
		\begin{definition}[Discrete-Time PE \cite{bai1985persistency}]
			\label{def:PE_DT}
			A signal $\ww\in \D{d}$ is Persistently Exciting if there exist positive scalars $\alpha$ and $T$ such that
			\begin{align}\label{eq:PE_dt}
				\sum_{\tau=t}^{t + T} w_\tau w_\tau^\top\geq \alpha I,
			\end{align}
			for all $t \in \N$. \oprocend
		\end{definition}
		
		\begin{definition}[Continuous-Time PE \cite{anderson1977exponential}]
			\label{def:PE_CT}
			A signal $\ww \in \C{d}$ is Persistently Exciting if there exist positive scalars $\alpha$ and $T$ such that
			\begin{align}\label{eq:PE_ct}
				\int_{t}^{t + T} w(\tau)w(\tau)^\top \textup{d}\tau \geq \alpha I,
			\end{align}
			for all $t \in \R_{\geq 0}$. \oprocend
		\end{definition}
		\begin{remark}
			In continuous-time, PE is typically defined by assuming some
			degree of smoothness of the involved signals
			\cite{morgan1977uniform,ioannou1996robust,
				anderson1977exponential}. Also boundedness is often
			assumed, in both continuous-time and discrete-time
			\cite{bitmead1984persistence}.  In this work, we restrict to
			the class of signals to $\C{d}$ and $\D{d}$ for simplicity,
			and we leave the analysis of non-smooth or unbounded signals
			for future developments.\oprocend
		\end{remark}	
		Similarly to \cite[Def. $2$]{narendra1987persistent}, define the subsets
		\begin{equation}
			\begin{split}
				\Dpe{d}&\coloneqq \{\ww\in\D{d}: \ww \text{ satisfies \eqref{eq:PE_dt}}\}\\
				\Cpe{d}&\coloneqq \{\ww\in\C{d}: \ww \text{ satisfies \eqref{eq:PE_ct}}\},
			\end{split}
		\end{equation}
		which are the sets of all the PE signals in the discrete- and continuous-time setting.
		The following lemma characterizes these sets.
		\begin{lemma}\label{lemma:PE_cone}
			$\Dpe{d}$ (resp., $\Cpe{d}$) is an open cone in $\D{d}$ (resp., $\C{d}$).\oprocend
		\end{lemma}
		The proof of Lemma~\ref{lemma:PE_cone} is provided in Appendix~\ref{proof:PE_cone}.
		\begin{remark}\label{remark:PE_robust}
			Notice that Lemma~\ref{lemma:PE_cone} implies robustness of the PE property 
			to sufficiently small signal perturbations (in a $\mathcal{L}_\infty$ sense), and it is a reformulation of \cite[Lemma $6.1.2$]{sastry1990adaptive}.
			In other words, being $\Cpe{d}$ an open cone, 
			for any a given signal $\ww \in \Cpe{d}$ we can always find a sufficiently close signal $\ww^\prime$ such that $\ww^\prime\in \Cpe{d}$. 
			However, this robustness is not the same for every PE signal.
			\oprocend
		\end{remark}

		\noindent
		Before addressing the main questions we want to answer, we clarify how the notion of ``sufficient richness" (SR) is intended in this work. 
		Consistently with the adaptive literature, we say an \emph{input} signal is sufficiently rich if, injected into a \emph{dynamical system}, it guarantees persistence of excitation of the \emph{output} signal we are interested in.
		Notice that imposing the PE to some output signal $\yy$ via input $\uu$ is a problem which is not in general guaranteed to be solvable. Furthermore, an input $\uu$ may be sufficiently rich for a specific system but not for a different one. 
		Specifically, the following two aspects determine the solvability (and the solution) of the problem:
		\begin{itemize}
			\item the structural properties of the dynamical system $\Sigma$ we are considering (maps $f$ and $h$);
			\item the set of initial conditions $x(0)$ of the systems which we are interested in.
		\end{itemize}
		These reasonings motivate the following formal definition for SR.	
		\begin{definition}[Sufficient Richness]\label{def:general_SR}
			Given a discrete-time (resp. continuous-time) dynamical system $\Sigma$ and initial condition $x_0$ (resp. $x(0)$)
			we say that the input signal $\uu\in\D{m}$ (resp. $\C{m}$) is sufficiently rich for $(\Sigma, x_0)$ if
			\begin{equation*}
				\yy = \Sigma(\uu, x_0) \in \Dpe{p} (\text{resp. } \Cpe{p}). \eqoprocend
			\end{equation*}
		\end{definition}
		Moreover, we denote the set of all SR input signals for system $\Sigma$ and initial condition $x_0 \in \R^n$  as
		\begin{equation}\label{eq:SR_sets}
			\begin{split}
				\Dsr\left(\Sigma, x_0\right) &\!\coloneqq\! \{\uu \!\in\! \D{m}: \Sigma(\uu, x_0)\in \Dpe{p}\}\\
				\Csr\left(\Sigma, x(0)\right) &\!\coloneqq\! \{\uu \!\in \!\C{m}: \Sigma(\uu, x(0))\in \Cpe{p}\}.
			\end{split}
		\end{equation} 
		\begin{remark}
			Notice that any characterization of the sets $\Csr(\cdot), \Dsr(\cdot)$ for the SR property as per Definition \ref{def:general_SR}, requires to explicitly state the considered system and initial conditions for which the characterization of SR is valid. \oprocend
		\end{remark}

	\subsection{Problem Statement} \label{sec:problem_statement}
	
		Consider the discrete- or continuous-time linear time invariant system of the form
		\begin{equation}\label{eq:plant_dynamics}
			\begin{split}
				x_{t+1} &= Ax_t + Bu_t, \hspace{1cm}  \dx(t)= Ax(t) + Bu(t),\\
				y_t & = Cx_t + Du_t \hspace{1.15cm} y(t) = Cx(t) + D u(t) 
			\end{split}
		\end{equation}
		where $x \in \R^n$ is the state, $u \in \R^m$ is the control input, 
		$y\in \R^p$ is the output, 
		$A\in \R^{\dimx\times \dimx}$ is the state matrix, $B \in \R^{\dimx\times \dimu}$ is the input matrix, while $C \in \R^{p \times \dimx}$ and $D \in \R^{p \times \dimu}$ are the output matrices.
		The aim of this article is twofold, namely, i) given meaningful classes of systems of the form \eqref{eq:plant_dynamics} and initial conditions, find an explicit characterization of the sets of SR inputs $\Csr(\cdot), \Dsr(\cdot)$; ii) obtain a characterization of $\Csr(\cdot), \Dsr(\cdot)$ which underlines the analogies between the discrete-time and the continuous-time framework.
		
		\noindent
		Concerning the first objective, we introduce now the classes of systems which will be investigated. First, since we are not interested in solving a control problem (and having defined PE only for bounded signals), we consider only asymptotically stable systems, i.e., systems such that $A$ is Schur (resp., Hurwitz). Intuitively, this assumption implies that the obtained characterizations for $\Csr(\cdot), \Dsr(\cdot)$ will not depend on the chosen initial condition.
		Second, given a stable system, it is known \cite[Lemma $5$]{narendra1987persistent}, \cite{ioannou1996robust} that a necessary condition for having PE state trajectories is that $(A, B)$ must be reachable. 
		Next, in this article we are interested in considering simple but significant output maps, namely, we consider the cases in which the output of a system is its state trajectory $\xx$ or the state-input stack $(\xx, \uu)$ (notice that PE of $(\xx, \uu)$ ensures PE of any surjective, constant linear map of $(\xx, \uu)$ \cite[Lemma $4.8.3$]{ioannou1996robust}).
		To include all the mentioned properties, define the classes of discrete-time systems
		\begin{equation}\label{eq:output_classes}
			\begin{split}
				\sctrbx &\coloneqq \{\Sigma  \text{ of the form \eqref{eq:plant_dynamics}}: A \text{ Schur}, (A, B) \text{ reachable}, C=I_n, D=0_m\}, \\
				\sctrbxu &\coloneqq \Big\{\Sigma  \text{ of the form \eqref{eq:plant_dynamics}}: A \text{ Schur}, (A, B) \text{ reachable},  C =
				\begin{bmatrix}
					I_n \\ 
					0_{m\times n} 
				\end{bmatrix}, 
				D=\begin{bmatrix}
					0_{n\times m}\\ I_m 
				\end{bmatrix}\Big\},
			\end{split}
		\end{equation}
		and their corresponding continuous-time counterparts $\hctrbx, \hctrbxu$.
		Given these definitions, when $\Sigma \in \sctrbx$ we denote with some abuse of notation $\Sigma(\uu, x_0)$ as $\xx$, and when $\Sigma \in \sctrbxu$ we denote $\Sigma(\uu, x_0)$ as $(\xx, \uu)$.
		Finally, we are interested in separating the results for single-input and multi-input systems, therefore we introduce the notation
		\begin{equation}\label{eq:simi_systems}
			\begin{split}
				&\sctrbxsi \!\!\coloneqq\! \{\Sigma \in\sctrbx : m=1\}, \hspace{0.1cm} \sctrbxusi \!\!\coloneqq\! \{\Sigma \!\in\!\sctrbxu : m=1\},\\
				&\sctrbxmi\! \!\coloneqq\! \{\Sigma \in\sctrbx : m>1\}, \hspace{0.1cm} \sctrbxumi \!\!\coloneqq\! \{\Sigma \!\in\!\sctrbxu : m>1\},
			\end{split}
		\end{equation}
		and their corresponding continuous-time counterparts $\hctrbxsi, \hctrbxmi, \hctrbxusi, \hctrbxumi$.
		\noindent
		Concerning the second objective, namely the problem of unifying the continuous and discrete-time results, our major concern is to avoid incompatible descriptions. 
		An example of such a language is the Hankel matrix notation, since it is not well-suited for the continuous-time framework. The idea is to rely instead on the correspondence between time shifts in discrete-time and derivatives in continuous-time, exploiting the operators defined in Section \ref{sec:problem_setup}.

\section{Main Results}\label{sec:main_result}
	
	\subsection{Discrete-Time Framework}
	
		\noindent
		We start now by giving necessary results for discrete-time systems. To the best of the authors' knowledge, this result is new, since necessary conditions have been recently found only for single-input systems \cite{markovsky2023persistency} (for the framework of finite-time excitation).
		Before presenting the theorem, we introduce a notion to characterize signals that persistently span only subspaces of the space they live in.  
		
		\begin{definition}[Discrete-Time Partial PE]
			\label{def:partial_PE_DT}
			A signal $\ww\in \D{d}$ is Partially Persistently Exciting (PPE) of degree 
			$d^\prime \leq d$ (and we write $\ww \in \Dppe{d}{d^\prime}$) if there exists a linear surjective map $P:\R^{d}\to \R^{d^\prime}$ such that $P\ww\in \Dpe{d^\prime}$.\oprocend
		\end{definition}

		\begin{remark}
			A similar notion of PPE was introduced for the continuous-time in \cite[Def. 3]{narendra1987persistent}, \cite[Def. $6.3$]{narendra2012stable}, from which we also adopt the simplified notation $\Dppe{d}{d^\prime}$. \oprocend
		\end{remark}

		\noindent
		It turns out that if the maximum degree $d^\prime$ of PPE of signals of the form $\QQ^n(\uu)$ (namely, stacks of time shifts of the same signal $\uu$, see Definition \ref{def:shift_operator}) is sufficiently small, then it never increases when increasing the number $n$ of time shifts.
		The intuition behind this property is that shifts of a certain time window of the same signal are not completely independent. Consider, e.g., the sequence of scalars $\{w_0, w_1, \ldots, w_n\}$ and its time shift $\{w_1, w_2, \ldots, w_{n+1}\}$. Notice that the two sequences share $n-1$ elements, and the same holds true also considering the following window $\{w_2, w_3, \ldots, w_{n+2}\}$. This constraint can be shown to bind the number of directions possibly spanned by the moving window $\{w_t, w_{t+1}, \ldots, w_{t+n-1}\}$ depending on the number of directions persistently spanned by the original signal $\ww$.
		This is formalized in the following lemma.

		\begin{lemma}\label{lemma:not_PE_in_DT}
			Let $\ww\in \D{d}$. Let $d^\prime \in \N$ be the biggest natural number such that
			$\QQ^n(\ww)\in \Dppe{nd}{d^\prime}$. If $d^\prime \leq d(n-1)$,  
			then for all $k\geq n$, the largest natural $d^{\prime\prime}$ such that $\QQ^{k} (\ww)\in \Dppe{k d}{d^{\prime\prime}}$ holds satisfies $d^{\prime\prime}\leq d^\prime$.\oprocend
		\end{lemma}
		The proof of Lemma~\ref{lemma:not_PE_in_DT} is provided in Appendix~\ref{proof:not_PE_in_DT}.

		\noindent
		We are now ready to state necessary conditions for PE in discrete-time LTI multivariable systems. 
		
		\begin{theorem}[Necessary condition for DT systems]\label{thm:nec_mi_dt}
			Let $\uu \in \D{m}$, for all $x_0\in \R^n$ the following holds:
			\begin{itemize}
				\item[i)] Consider a discrete-time asymptotically stable, reachable linear system $\Sigma$ with persistently exciting output $\xx$. Then, the signal $\QQ^n(\uu)$ is partially persistently exciting of degree $n$, i.e.,  
				\begin{equation}
					\Sigma \in \sctrbx,\; \xx\in \Dpe{n} \implies \QQ^n(\uu)\in \Dppe{nm}{n}.
				\end{equation}
				\item[ii)] Consider a discrete-time asymptotically stable, reachable linear system $\Sigma$ with persistently exciting output $(\xx, \uu)$. Then, the signal $\QQ^{n+1}(\uu)$ is partially persistently exciting of degree $n+m$, i.e., 
				\begin{equation}
					\Sigma\! \in \!\sctrbxu,\! (\xx, \uu)\!\in \!\Dpe{n+m} \!\!\implies\!\! \QQ^{n+1}\!(\uu)\!\in\!\Dppe{(n+1)m}{n+m}.
				\end{equation}
			\end{itemize}
			\oprocend
		\end{theorem}
		The proof of Theorem~\ref{thm:nec_mi_dt} is provided in Appendix~\ref{proof:nec_mi_dt}.
		
		\noindent
		We summarize here its main steps,
		since they provide arguments analogous to those in the continuous-time framework and give insight on why PPE of $\QQ^n(\uu)$ is a necessary condition. In detail, Theorem~\ref{thm:nec_mi_dt}
		is proven by contraposition (cf. Appendix~\ref{proof:nec_mi_dt}), namely, we show that if the input signal $\uu$ is not PPE of a sufficiently high degree then the resulting $\xx$ is not PE. To this end, it uses the following arguments.
		\begin{enumerate}
			\item In asymptotically stable systems, $\xx$ can be approximated arbitrarily well by a linear function of a large enough number $k\in \N$ of time-shifts of the input $\QQ^k(\uu)$, namely, $\xx \approx K\QQ^k(\uu)$.
			\item If $\QQ^n(\uu)$ is PPE of degree at most $n-1$, then for any $k\geq n$, the PPE of $\QQ^k(\uu)$ does not increase (cf. Lemma~\ref{lemma:not_PE_in_DT}).
			\item If $\QQ^k(\uu)$ is PPE of degree at most $n-1$, then $\xx$ cannot be PE, since it is a linear function of a signal persistently spanning at most $n-1$ directions.
		\end{enumerate}
		\begin{remark}
			Theorem~\ref{thm:nec_mi_dt} assumes
			$A$ Schur. 
			This assumption is required both (i) to guarantee boundedness of the state and (ii) to approximate the state as a linear function of a finite number of inputs.\oprocend
		\end{remark}
		\begin{remark}\label{remark:necessary_for_SI_DT}
			Notice that, for single-input systems, Theorem~\ref{thm:nec_mi_dt} collapses to a ``full" persistency of excitation requirement on $\QQ^n(\uu)$ and $\QQ^{n+1}(\uu)$, namely, for all $x_0\in \R^n$ it holds
			\begin{itemize}
				\item[i)] $\Sigma \in \sctrbxsi, \xx\in \Dpe{n} \implies \QQ^n(\uu)\in \Dpe{n}$.
				\item[ii)] $\Sigma \in \sctrbxusi, (\xx, \uu)\!\in \!\Dpe{n+m} \implies \QQ^{n+1}(\uu)\in\Dpe{n+1}$.\oprocend
			\end{itemize}
		\end{remark}

		\noindent 
		We now proceed by finding sufficient conditions for discrete-time systems. While the following results are known in the literature \cite{moore1983persistence, bai1985persistency, green1986persistence}, we note that the proof given here is new, provides insight on the obtained results, and its principles paves the way for obtaining them in the continuous-time framework too.
		
		\begin{theorem}[Sufficient condition for DT systems]\label{thm:suff_mi_dt}
			Let $\uu \in \D{m}$, for all $x_0\in \R^n$ the following holds:
			\begin{itemize}
				\item[i)] Consider a discrete-time asymptotically stable, reachable linear system $\Sigma$ with output $\xx$. If the signal $\QQ^{n}(\uu)$ is persistently exciting, then $\xx$ is persistently exciting, i.e., 
				\begin{equation*}
					\Sigma \in \sctrbx, \QQ^n(\uu) \in \Dpe{nm} \implies \xx\in \Dpe{n}.
				\end{equation*}
				\item[i)] Consider a discrete-time asymptotically stable, reachable linear system $\Sigma$ with output $(\xx, \uu)$. If the signal $\QQ^{n+1}(\uu)$ is persistently exciting, then $(\xx, \uu)$ is persistently exciting, i.e., 
				\begin{equation*}
					\Sigma \in \sctrbxu, \QQ^{n+1}(\uu) \! \in \! \Dpe{(n+1)m} \! \implies \! (\xx, \uu)\in \Dpe{n+m}.\eqoprocend
				\end{equation*}
			\end{itemize}
		\end{theorem}
		The proof of Theorem~\ref{thm:suff_mi_dt} is provided in Appendix~\ref{proof:suff_mi_dt}.
		
		\noindent
		We summarize as follows its main principles, starting, for simplicity from the single-input case.
		The proof of Theorem~\ref{thm:suff_mi_dt} (cf. Appendix~\ref{proof:suff_mi_dt}) proceeds by contraposition 
		according to the following points.
		\begin{enumerate}
			\item Any solution $\xx=\Sigma(\uu, x(0))\in \D{n}$ of system $\Sigma$ which is not PE constrains the (scalar) system input to be arbitrarily close to a feedback gain for arbitrarily long intervals, namely, $u_t\approx Kx_t$.
			\item Such an input makes the system an autonomous system for arbitrarily long intervals. Hence, time shifts $x_{t}, x_{t-1}, \ldots, x_{t-n+1}$ of the state can be written as a linear function of the state $x_{t-n+1}$.
			\item Time shifts of inputs are given by state feedbacks of time shifts of the state, namely $(u_t, \ldots, u_{t-n+1}) \approx (Kx_t, \ldots, Kx_{t-n+1})$, which are a linear function of the state $x_{t-n+1}$. This means that if $\xx$ is not PE, $\QQ^n(\uu)$ is not PE. 
		\end{enumerate}
		The principles behind the proof for multi-input systems are the same as those for single-input systems, with the differences explained in the following steps.
		\begin{enumerate}
			\item $\xx$ not being PE does not fully constrain the input signal to be arbitrarily close to a feedback gain of the state. In particular, we obtain an input which can be written as $u_t \approx Kx_t + v_t$.
			\item However, $v_t$ is constrained to span in a space which is at most $m-1$ dimensional.
			\item It can be shown that this means $\QQ^n(\uu)$ is not PE, since it is a function of $\xx$ and $\QQ^n(\vv)$ which, altogether, span persistently a subspace of $\R^{nm}$ which is at most $nm-1$ dimensional. 
		\end{enumerate}
		\begin{remark}
			Theorem~\ref{thm:suff_mi_dt} assumes that the state matrix $A$ is Schur. We require this assumption only to guarantee boundedness of the state (cf. Definition~\ref{def:PE_DT}). Hence, in the case of PE definitions without the boundedness requirement (see, e.g., \cite[Def.~3]{nordstrom1987persistency}), the Schur property of $A$ is not required. \oprocend
		\end{remark}

	\subsection{Continuous-Time Framework}
		
		Before presenting the results for the continuous-time framework, we need to consider some problems that arise when applying Definition \ref{def:PE_CT}, which are not present in the discrete-time framework. With the next result, we leverage the smoothness properties of the considered signals to guarantee that a non-PE signal must be arbitrarily small for arbitrarily long periods along certain directions. 
		\begin{lemma}\label{lemma:CT_not_disagiato}
			Let $\ww \in \C{d}$. If $\ww \notin \Cpe{d}$, then for any $T, \epsilon>0$ there exist $t>0$ unitary $z \in \R^d$ such that
			\begin{equation}
				|z^\top w(\tau)|\leq \epsilon
			\end{equation} 
			for all $\tau \in [t, t+T]$.
			\oprocend
		\end{lemma}
		The proof of Lemma~\ref{lemma:CT_not_disagiato} is provided in Appendix~\ref{proof:CT_not_disagiato}.
		
		\noindent 
		Notice that while the above result is immediate in the discrete-time framework, in continuous-time it is not necessarily true without assuming some degrees of smoothness of the signal.
		Next, we introduce the notion of PPE also for this framework.
		
		\begin{definition}[Continuous-Time Partial PE]
			\label{def:partial_PE_CT}
			A signal $\ww\in C^d$ is Partially Persistently Exciting (PPE) of degree $d^\prime \leq d$ (and we write $\ww \in \Cppe{d}{d^\prime}$) if there exists a linear surjective map $P:\R^{d}\to \R^{d^\prime}$ such that $P\ww\in \Cpe{d^\prime}$.\oprocend
		\end{definition}
		
		\noindent To get the necessary conditions of PE in the continuous-time framework by following the same steps of the discrete-time one, we need a result analogous to Lemma \ref{lemma:not_PE_in_DT}.

		\begin{lemma}\label{lemma:not_PE_in_CT}
			Let $\ww \in \C{d}$. Let $d^\prime \in \N$ be the largest natural number for which $\WW\coloneqq\DD^n(\ww)\in \Cppe{nd}{d^\prime}$.
			If $d^\prime \leq d(n-1)$, then
			for any $T, \bar{T}, \epsilon>0$ such that $\bar{T}\leq  T$, there exist $t>0, \bar{N}\in \N$, $d^{\prime \prime}\leq d^\prime$, and $G\in \R^{nd(N+1)\times d^{\prime\prime}}$ such that, for all $N\geq \bar{N}$,
			\begin{equation}
				\begin{bmatrix}
					W(\tau)\\
					W\left(\tau-\frac{\bar{T}}{N}\right)\\
					\vdots\\
					W(\tau-\bar{T})
				\end{bmatrix}
				= G\lambda(\tau) + \tilde{W}(\tau), \;\;\; \forall \tau \in [t, t+T],
			\end{equation}
			for some 
			$\lambda(\tau)\in \R^{d^{\prime\prime}}$ and $\tilde{W}(\tau)\in \R^{nd(N+1)}$ such that $|\tilde{W}(\tau)|\leq \epsilon$. \oprocend
		\end{lemma}
		The proof of Lemma~\ref{lemma:not_PE_in_CT} is provided in Appendix~\ref{proof:not_PE_in_CT}.

		\noindent
		The meaning of Lemma~\ref{lemma:not_PE_in_CT} is analogous to its discrete-time counterpart (i.e., Lemma~\ref{lemma:not_PE_in_DT}): under a certain threshold of PPE, and sampling with a sufficiently small sample time, the directions spanned persistently by a certain stack of time shifts of $\DD^n(\ww)$ do not increase if the time shifts are increased.\\
		Given Lemma~\ref{lemma:not_PE_in_CT}, and following the same steps of the discrete-time case, we obtain the following necessary conditions.

		\begin{theorem}[Necessary condition for CT systems]\label{thm:nec_mi_ct}
			Let $\uu \in \C{m}$. For all initial conditions $x(0)\in \R^n$,
			\begin{itemize}
				\item[i)] $\Sigma\in \hctrbx, \xx \in \Cpe{n} \implies \DD^n(\uu)\in \Cppe{nm}{n}$.
				\item[ii)] $\Sigma\in \hctrbxu, (\xx, \uu)\in \Cpe{n+m} \!\implies \!\DD^{n+1}(\uu)\!\in\! \Cppe{(n+1)m}{n+m}$.
			\end{itemize}
			\oprocend
		\end{theorem}

		The proof of Theorem~\ref{thm:nec_mi_ct} is provided in Appendix~\ref{proof:nec_mi_ct}.

		\begin{remark}\label{remark:necessary_for_SI_CT}
			As in the case of discrete-time systems, we note that, when $m=1$, Theorem~\ref{thm:nec_mi_ct} collapses into a ``full" persistency of excitation requirement on $\DD^n(\uu)$ and $\DD^{n+1}(\uu)$, namely, for all $x(0)\in \R^n$, it holds
			\begin{itemize}
				\item[i)]$\Sigma\in \hctrbxsi, \xx\in \Cpe{n}\implies\QQ^n(\uu)\in \Cpe{n}$
				\item[ii)]$\Sigma\in \hctrbxusi, (\xx, \uu)\in \Cpe{n+1}\implies\QQ^{n+1}(\uu)\in \ \Cpe{n+1}.$\oprocend
			\end{itemize}
		\end{remark}

		\noindent
		Moving to sufficient conditions in continuous-time systems, 
		in order to repeat the proof given for the sufficient results in the discrete-time domain, we now provide a relation between the PE of a signal and the PE of its time derivative.  
		
		\begin{lemma}\label{lemma:PE_of_derivatives}
			Let $\ww \in \C{d}$. If $\ww \notin \Cpe{d}$, then for any $T, \epsilon>0$ there exist $t>0$ and unitary $z \in \R^d$ such that
			\begin{align}
				|z^\top w(\tau)|\leq \epsilon, &&& |y^\top \dot{w}(\tau)|\leq \epsilon, 
			\end{align} 
			for all $\tau \in [t, t+T]$.\oprocend
		\end{lemma}
		The proof of Lemma~\ref{lemma:PE_of_derivatives} is provided in Appendix~\ref{proof:PE_of_derivatives}.
		
		\noindent
		With Lemma~\ref{lemma:CT_not_disagiato} and~\ref{lemma:PE_of_derivatives} at hand, we can adapt the proof of Theorem~\ref{thm:suff_mi_dt} to derive an equivalent result for continuous-time systems.

		\begin{theorem}[Sufficient condition for CT systems]\label{thm:suff_mi_ct}
			Let $\uu \in \C{m}$. For all initial conditions $x(0)\in \R^n$,
			\begin{itemize}
				\item[i)] $\Sigma\in \hctrbx, \DD^{n}(\uu) \in \Cpe{nm} \implies \xx \in \Cpe{n}$.
				\item[ii)] $\Sigma\in \hctrbxu, \DD^{n+1}(\uu) \in \Cpe{(n+1)m} \implies (\xx, \uu)\in \Cpe{n+m}$.\oprocend
			\end{itemize}
		\end{theorem}
		The proof of Theorem~\ref{thm:suff_mi_ct} is provided in Appendix~\ref{proof:suff_mi_ct}.

	\subsection{A discussion on the sets of SR signals}
		
		Having collected all these necessary and sufficient conditions, we are now interested in explictly characterizing the set of the signals which are SR \emph{for all} stable systems sharing certain structural properties, namely, the input and state dimensions.
		At first, we consider the case of single input systems, i.e., we are interested in characterizations of the sets 
		\begin{align}\label{eq:SR_si}
			\Dsr(\sctrbxsi)\coloneqq \!\!\bigcap_{\substack{x_0\in \R^n\\ \Sigma \in \sctrbxsi}}\!\!\Dsr(\Sigma, x_0), &&& \Csr(\hctrbxsi)\coloneqq\!\! \bigcap_{\substack{x(0)\in \R^n\\ \Sigma \in \hctrbxsi}} \!\! \Csr(\Sigma, x_0),
		\end{align}
		
		\noindent
		where $\Dsr(\Sigma, x_0)$ and $\Csr(\Sigma, x(0))$ are given in \eqref{eq:SR_sets}.
		
		\begin{lemma}
			Given system classes $\sctrbxsi, \hctrbxsi$ in \eqref{eq:simi_systems}, the sets of sufficiently rich input in \eqref{eq:SR_si} are given by
			\begin{equation}
				\begin{split}
					\Dsr(\sctrbxsi)	&= \{\uu \in \D{}: \QQ^n(\uu) \in \Dpe{n}\}, \\
					\Csr(\hctrbxsi)	&= \{\uu \in \C{}: \DD^n(\uu) \in \Cpe{n}\}.
				\end{split}
			\end{equation}
		\end{lemma}
		\begin{proof}
			It is sufficient to apply Theorems \ref{thm:nec_mi_dt}, \ref{thm:suff_mi_dt}, \ref{thm:nec_mi_ct}, and \ref{thm:suff_mi_ct} to notice that
			\begin{equation}
				\begin{split}
					\QQ^n(\uu)\in \Dpe{n} \iff \Sigma(\uu, x_0)\in \Dpe{n},
				\end{split}
			\end{equation}
			for all $\Sigma \in \sctrbxsi$ and $x_0\in \R^n$.
		\end{proof}
		\begin{remark}
			Notice that the sets $\Dsr(\sctrbxsi), \Csr(\hctrbxsi)$ are open cones in $\D{n}, \C{n}$ (the proof is the same as the one for Lemma~\ref{lemma:PE_cone}). As per Remark \ref{remark:PE_robust}, this means that sufficiently small perturbations of SR signals are still SR signals. \oprocend
		\end{remark}
		This characterization is complete, namely, we have found \emph{all} the input signals which are SR for systems in $\sctrbxsi, \hctrbxsi$. 
		Indeed, the single-input case is simplified by the fact that, as stated in \cite[Thm. $1$]{narendra1987persistent}, any signal which is SR for a certain $\Sigma_1 \in \sctrbx$ must be SR also for $\Sigma_2 \in \sctrbx$ (which means that in \eqref{eq:SR_si} we intersect always the same set). The same is not true for multi-input systems,
		and in this case a complete characterization of the inputs which are SR for \emph{all} the stable systems sharing input and state dimension seems not easy to obtain. Considering the classes $\sctrbxmi, \hctrbxmi$ defined in \eqref{eq:simi_systems}, 
		we are interested in the sets
		\begin{align}\label{eq:SR_mi}
			\Dsr(\sctrbxmi)\coloneqq \!\!\!\bigcap_{\substack{x_0\in \R^n\\ \Sigma \in \sctrbxmi}}\!\!\!\Dsr(\Sigma, x_0), &&& \Csr(\hctrbxmi)\coloneqq \!\!\!\bigcap_{\substack{x(0)\in \R^n\\ \Sigma \in \hctrbxmi}}\!\!\!\Csr(\Sigma, x_0),
		\end{align}
		
		\noindent
		where $\Dsr(\Sigma, x_0)$ and $\Csr(\Sigma, x(0))$ are given in~\eqref{eq:SR_sets}.
		\begin{lemma}
			Given system classes $\sctrbxmi, \hctrbxmi$ in \eqref{eq:simi_systems}, the sets of sufficiently rich input \eqref{eq:SR_mi} satisfy
			\begin{equation}
				\begin{split}
					&\{\uu\!:\!\QQ^n(\uu)\!\in \!\Dpe{nm}\}\! \subset  \! \Dsr(\sctrbxmi) \! \subset\! \{\uu\!:\!\QQ^n(\uu)\!\in\! \Dppe{nm}{n}\}
					\\
					&\{\uu\!:\!\DD^n(\uu)\!\in \!\Cpe{nm}\} \!\subset \!  \Csr(\hctrbxmi) \! \subset \!\{\uu\!:\!\DD^n(\uu)\!\in\! \Cppe{nm}{n}\}.\!\!
				\end{split}
			\end{equation}
		\end{lemma}
		\begin{proof}
			It is sufficient to apply Theorems \ref{thm:nec_mi_dt}, \ref{thm:suff_mi_dt}, \ref{thm:nec_mi_ct}, \ref{thm:suff_mi_ct} to notice that
			\begin{equation}
				\begin{split}
					\QQ^n(\uu)\in \Dpe{nm} &\implies \Sigma(\uu, x_0)\in \Dpe{n} 
					\\
					\QQ^n(\uu)\in \Dppe{nm}{n} &\impliedby \Sigma(\uu, x_0)\in \Dpe{n},
				\end{split}
			\end{equation}
			for all $\Sigma \in \sctrbxmi$ and $x_0\in \R^n$.
		\end{proof}
		The difficulties of obtaining a complete characterization of the sets $\Dsr(\sctrbxmi),  \Csr(\hctrbxmi)$ are due by the fact that in \eqref{eq:SR_mi} we are intersecting different sets (in other words, there exist inputs which are SR for a certain $\Sigma_1 \in \sctrbxmi$ but not for $\Sigma_2 \in \sctrbxmi$).
		The reason of this needs not to be searched into the initial condition, but into the system matrices $A, B$ and how each input enters into the system. 
		To corroborate this statement notice that, by \cite[Lemma~2.2]{Wonham74}, if $(A, B)$ is reachable, then it is possible to build a feedback gain $F$ such that $(A+BF, b)$ is reachable, with $b\in \im(B)$. With such a reshape of the system, the conditions for single input systems hold, which are independent on $A, b, F$, and $x_0$ (so, it is only the \emph{structure} of the system that influences the ability to obtain PE trajectories).

\section{Counterexamples}\label{sec:example}
	
	We conclude the article by providing examples which show that the obtained results are tight and cannot be improved without considering more specific classes of systems. 
	In both the following examples, we inject periodic inputs and gather sufficiently large data batches to obtain a representative behavior for the whole infinite-dimensional time window.

	\subsection{Sufficient condition}
	
		We pick here a condition about $\uu$ weaker than the one claimed to be sufficient by Theorem~\ref{thm:suff_mi_dt}, and we show it is not enough to guarantee PE of $(\xx, \uu)$. We consider a discrete-time LTI system in form \eqref{eq:plant_dynamics}, with $n=7, m=3$ and matrices 	
		\begin{equation}\label{eq:sufficiency_counterexample}
			\begin{split}
				A&\coloneqq
				\begin{bmatrix}
					0 &   1  &       0   &      0   &      0   &      0     &    0\\
					0 &        0  &  1    &     0    &     0    &     0    &     0\\
					0.024 &  -0.26  &  0.9    &     0     &    0      &   0     &    0\\
					0  &       0   &      0      &   0  &   1     &    0      &   0\\
					0  &       0   &      0     &    0    &     0   &  1      &   0\\
					0  &       0   &      0 &   0.21 &  -1.07 &   1.8   &      0\\
					0  &       0   &      0   &      0    &     0     &    0  &   0.8
				\end{bmatrix},\\
				B& \coloneqq 
				\begin{bmatrix}
					0  &  2 &   1   &      0   &      0     &    0 &   1\\
					2 &   1  &  0.4  &  7  &  4     &    0   &	0  \\
					5  &  2  &  0.9  &  4  &  6  &  2	&	1
				\end{bmatrix}^\top.
			\end{split}
		\end{equation}
		
		\noindent
		It can be verified that	the pair $(A, B)$ is reachable and $A$ is Schur.
		We want to show that $\QQ^{n}(\uu)\in \Dpe{nm}$ (instead of $\QQ^{n+1}(\uu)\in \Dpe{(n+1)m}$) does not ensure $(\xx, \uu)\in \Cpe{n+m}$.
		By choosing initial condition $x_0=0$, and input described by the dynamics $u_{t+1} = K_xx_t + K_u u_t + v_t^1 + v_t^2$, with $u_0=0$ and
		\begin{equation*}
			\begin{split}
				v_t^1 &= \begin{bmatrix}
					-0.4082\\
					0.9082\\
					0.0918
				\end{bmatrix} (\sin(t)+\sin(2t)+\sin(3t)+\sin(4t)),\\
				v_t^2 &= 
				\begin{bmatrix}
					0.4082\\
					0.0918\\
					0.9082
				\end{bmatrix} (\sin(5t)+\sin(6t)+\sin(7t)+\sin(8t)),\\
				K_x &= 10^{-3}
				\setlength{\arraycolsep}{2pt}
				\begin{bmatrix}
					-8  &  8.67  & -300 &  -70  &  356.7  & -600  & -266.7\\
					-4  &  4.33  & -150 &  -350  &  178.3 &  -300  & -133.3\\
					4  & -4.33  &  150 &   350 &  -178.3   & 300 &   133.3
				\end{bmatrix},\\
				\setlength{\arraycolsep}{5pt}
				K_u &= 
				\begin{bmatrix}
					-0.6667 &  -0.1333  & -1.3\\
					-0.3333 &  -0.0667 &  -0.65\\
					0.3333  &  0.0667  &  0.65
				\end{bmatrix},
			\end{split}
		\end{equation*} 
		simulating for $t=1, \ldots, 1000$, it can be checked that the input verifies $\QQ^{n}(\uu) \in \Dpe{nm}$. However the resulting state-input trajectory $(\xx, \uu)$ is not PE, so we have shown that 
		\begin{equation}
			\QQ^{n}(\uu) \in \Dpe{nm}\;\; \not \hspace{-0.25cm}\implies \;\;(\xx, \uu)\in \Dpe{n+m}.
		\end{equation}
		To corroborate these claims, in Figure \ref{fig:rankssufficiency}, we plot the quantities
		\begin{equation}\label{eq:r_definitions}
			\begin{split}
				r_{1}(\xx, \uu, T)\coloneqq& \rank\left(\sum_{t=0}^T (x_t, u_t)(x_t, u_t)^\top\right)
				\\
				r_n(\uu, T)\coloneqq &\rank\left(\sum_{t=0}^T \QQ^n(\uu)_t \QQ^n(\uu)_t^{\top}\right).
			\end{split}
		\end{equation}
		\begin{figure}[h!]
			\centering
			\includegraphics[width=0.7\linewidth]{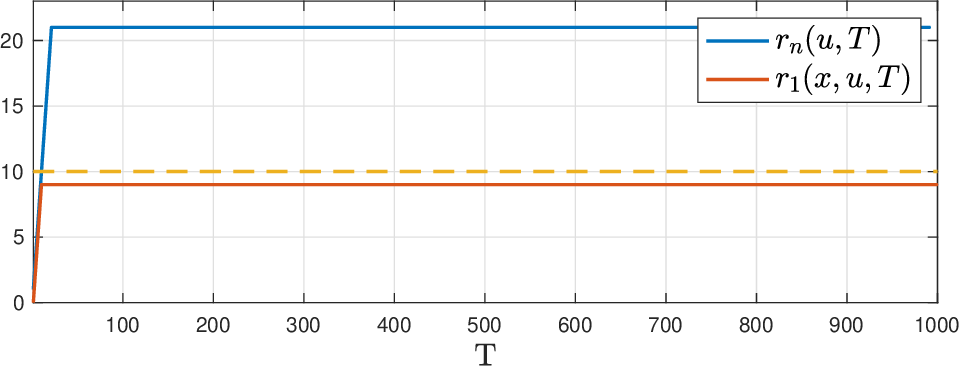}
			\caption{Directions spanned in time by the signals $(\xx, \uu)$ and $\QQ^{n}(\uu)$.}
			\label{fig:rankssufficiency}
		\end{figure}

		\begin{remark}
			This example is a counterexample also for the sufficiency conjecture in \cite[Pag.~4]{markovsky2023persistency}, which we may state (with some abuse of notation) as 
			\begin{equation}
				\QQ^{\nu+1}(\uu) \in \Dpe{(\nu+1)m} \implies (\xx, \uu)\in \Dpe{n+m},
			\end{equation}
			where $\nu$ is the controllability index \cite[Pag.~121]{Wonham74} of the pair $(A, B)$. In this case, the controllability index of the pair \eqref{eq:sufficiency_counterexample} is $\nu=3$, and $\nu + 1 < n$, so we may write $\QQ^{\nu+1}(\uu) = P \QQ^{n}(\uu)$ for some full row rank matrix $P$, and $\QQ^{n}(\uu) \in \Dpe{nm}$ ensures that our input verifies also $\QQ^{\nu+1}(\uu)\in \Dpe{(\nu+1)m}$ \cite[Lemma~4.8.3]{ioannou1996robust}.
			However, we have shown that $(\xx, \uu)\in \Dpe{n+m}$ is not achieved by this input, thus this counterexample demonstrates also that 
			\begin{equation*}
				\QQ^{\nu+1}(\uu) \in \Dpe{(\nu+1)m}\;\; \not \hspace{-0.25cm}\implies \;\;(\xx, \uu)\in \Dpe{n+m}.\eqoprocend
			\end{equation*}
		\end{remark}

	\subsection{Necessary condition}
	
		Here, we focus on a condition about $\uu$ stronger than the one claimed to be necessary by Theorem~\ref{thm:nec_mi_dt}, and we show it is not guaranteed by PE of $(\xx, \uu)$. We consider a discrete-time LTI system in form \eqref{eq:plant_dynamics}, with $n=7, m=3$, matrix $B$ given in \eqref{eq:sufficiency_counterexample}, and
		\begin{equation*}
			\begin{split}
				A&\coloneqq
				\begin{bmatrix}
					0   & 1    &     0    &     0    &     0    &     0   &      0\\
					0   &      0  &  1    &     0    &     0    &     0    &     0\\
					-0.3 &   0.2  &  0.1   &      0   &      0      &   0    &     0\\
					0   &      0    &     0    &     0  &  1    &     0    &     0\\
					0    &     0   &      0   &      0    &     0   & 1    &     0\\
					0     &    0  &       0  & -0.3  &  0.2  &  0.1    &     0\\
					0     &    0 &        0    &     0   &      0    &     0  & -0.7324
				\end{bmatrix}.
			\end{split}
		\end{equation*}
		
		\noindent
		It can be verified that	the pair $(A, B)$ is reachable and $A$ is Schur.
		We want to show that $(\xx, \uu)\in \Dpe{n+m}$ does not ensure $\QQ^{n+1}(\uu) \in \Dppe{(n+1)m}{n+m+1}$.
		By choosing initial condition $x_0=0$, and input 
		\begin{equation}
			u_t = 
			\begin{bmatrix}
				\sin(t)+\sin(2t)\\
				\sin(3t) + \sin(4t)\\
				\sin(5t)
			\end{bmatrix},
		\end{equation}
		simulating for $t=1, \ldots, 1000$, it can be verified that $(\xx, \uu)\in \Dpe{n+m}$. However, it can be verified also that $\QQ^{n+1}(\uu) \notin \Dpe{(n+1)m, n+m+1}$.
		With this example, we have shown that
		\begin{equation}
			(\xx, \uu)\in \Dpe{n+m} \;\; \not \hspace{-0.25cm}\implies \;\;\QQ^{n+1}(\uu) \in \Dppe{(n+1)m}{n+m+1}.
		\end{equation}
		To corroborate these claims, in Figure \ref{fig:rankssufficiency}, we plot the quantities $r_{1}(x, u, T)$ as defined in \eqref{eq:r_definitions} and
		\begin{equation}
			\begin{split}
				r_{n+1}(\uu, T)\coloneqq &\rank\left(\sum_{t=0}^T \QQ^{n+1}(\uu)_t \QQ^{n+1}(\uu)_t^{\top}\right).
			\end{split}
		\end{equation}
		Notice these plots show that the  necessary condition given in Theorem \ref{thm:nec_mi_dt} holds.
		
		\begin{figure}[h!]
			\centering
			\includegraphics[width=0.7\linewidth]{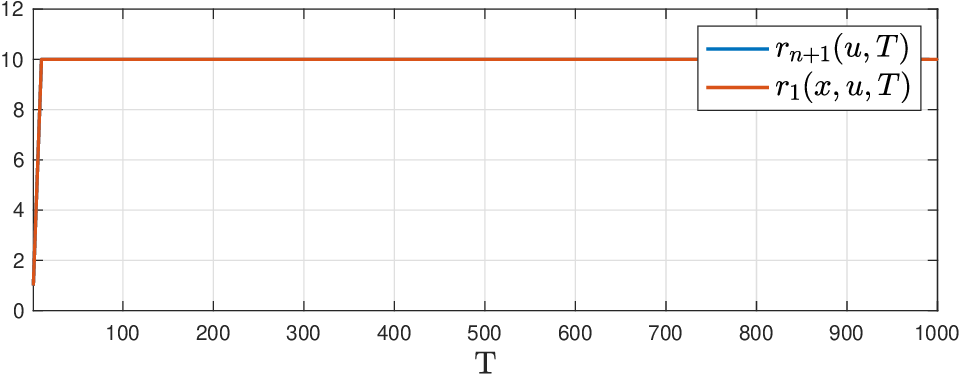}
			\caption{Directions spanned in time by the signals $(\xx, \uu)$ and $\QQ^{n+1}(\uu)$.}
			\label{fig:ranksnecessity}
		\end{figure}

		\begin{remark}
			This example is a counterexample also for the necessary conjecture in \cite[Pag. $4$]{markovsky2023persistency}, which we may state (with some abuse of notation) as 
			\begin{equation}
				(\xx, \uu)\in \Dpe{n+m} \implies \QQ^{\nu+1}(\uu) \in \Dpe{(\nu+1)m},
			\end{equation}
			where $\nu$ is the controllability index \cite[Pag. $121$]{Wonham74} of the pair $(A, B)$. In this case, the controllability index of the pair \eqref{eq:sufficiency_counterexample} is $\nu=3$, 
			and $\nu + 1 < n$, so we may write $\QQ^{\nu+1}(\uu) = P \QQ^{n}(\uu)$ for some full row rank matrix $P$. However, since we have shown that $ \QQ^{n}(\uu)$ spans only $n+m=10$ directions, $\QQ^{\nu+1}(\uu)$ can span at most $10$ directions. Being $\QQ^{\nu+1}(\uu)$ $(\nu+1)m=12$-dimensional, this means it is not PE, and thus this counterexample demonstrates also that 
			\begin{equation*}
				(\xx, \uu)\in \Dpe{n+m} \;\; \not \hspace{-0.25cm}\implies \;\;\QQ^{\nu+1}(\uu) \in \Dpe{(\nu+1)m}.\eqoprocend
			\end{equation*}
		\end{remark}

\section{Conclusions}\label{sec:conclusion}

	In this paper, we have addressed the problem of guaranteeing the persistence of excitation of state and input signals in the context of LTI systems via the application of a sufficiently rich input. 
	
	Exploiting the analogies between time shifts in discrete time and derivatives in continuous time, we are able to develop a unifying notation to state necessary and sufficient conditions to obtain PE of commonly used regressors for both frameworks.
	Leveraging these conditions, we explicitly characterized the set of sufficiently rich signals for asymptotically stable controllable LTI systems.
	Finally, we have shown with a numerical example that the derived conditions are tight and cannot be improved without including more specific knowledge of the considered system.

\section{Appendix}

	\subsection{Proof of Lemma \ref{lemma:PE_cone}}\label{proof:PE_cone}
	
	Since the arguments for discrete-time are analogous but more straightforward, we prove only the continuous-time part.
	Notice that for any $\ww\in \Cpe{d}$ and $\lambda>0$, $\lambda \ww \in \Cpe{d}$, so $\Cpe{d}$ is a cone in $\C{d}$.
	Next, we show it is open. Let $\ww \in \Cpe{d}$. There exist $T, \alpha>0$ such that 
	\begin{equation}
		\int_t^{t+T} w(\tau)w(\tau)^\top d\tau\geq \alpha I, \;\;\; \forall t\geq 0.
	\end{equation}
	Choose any $\alpha^\prime \in (0, \alpha)$ and any 
	$\epsilon \in (0, \frac{\alpha^\prime}{2TM})$
	, where $M\coloneqq \|\ww\|_\infty$. Choose any $\ww^\prime$ such that $\|\ww^\prime-\ww\|_\infty=\|\Delta \ww\|_\infty\leq \epsilon$. We have
	\begin{equation}
		\begin{split}
		\int_t^{t+T} w(\tau)^\prime w(\tau)^{\prime\top} \textup{d}\tau =&\int_t^{t+T} (w(\tau) + \Delta w(\tau ))(w(\tau)+\Delta w(\tau))^\top \textup{d}\tau\\
		\geq& \alpha I + \int_t^{t+T} \Big(w(\tau)\Delta w(\tau)^\top+\Delta w(\tau)w(\tau)^\top
		+ \Delta w(\tau)\Delta w(\tau)^\top \Big)\textup{d}\tau\\
		\geq& \alpha I + \int_t^{t+T} \Big(w(\tau)\Delta w(\tau)^\top+\Delta w(\tau)w(\tau)^\top\Big) \textup{d}\tau\\
		\geq &I\left(\alpha - \int_t^{t+T} 2M \epsilon \textup{d}\tau \right)\\
		\geq &I(\alpha - 2TM \epsilon)\geq(\alpha-\alpha^\prime) I>0,
		\end{split}
	\end{equation} 
	so, $\ww^\prime$ is PE.
	Therefore, for each point $\ww\in \Cpe{d}$, it is always possible to find an open ball about $\ww$ which is still in $\Cpe{d}$.

	\subsection{Proof of Lemma \ref{lemma:not_PE_in_DT}}\label{proof:not_PE_in_DT}
	
	We divide the proof in three steps.
	
	\noindent 
	I) \emph{A useful characterization for PPE signals}
	
	Since $\WW(\ww)\coloneqq\QQ^{n}(\ww)$ is PPE of degree at most $d^\prime\leq d(n-1)$, 
	there exist $i=1, \ldots, nd-d^\prime$ orthonormal directions $z_i\in \R^{nd}$ such that for each $T\in \N, \epsilon>0$ we can find $\tet\in \N$ such that
	\begin{equation}\label{eq:claim}
		\sum_{\tau=\tet}^{\tet+T}|W^\top_\tau z_i|\leq \epsilon,
	\end{equation}
	which means $|W^\top_\tau z_i|\leq \epsilon$, for all $\tau=\tet, \ldots, \tet+T$. 
	Consider an orthonormal basis for $\R^{nd}$, $\{z_1, \ldots, z_{nd}\}$.
	Since $W_\tau = \sum_{i=1}^{nd} z_iW^\top_\tau z_i$, there exists $\lambda_\tau \in \R^{d^\prime}$ such that
	\begin{equation}\label{eq:PPE_implies_linear_dependence}
		W_\tau = E \lambda_\tau + \tilde{W}_\tau
	\end{equation}
	for all $\tau=\tet, \ldots, \tet+T$, where $E=[z_{nd-d^\prime + 1}, \ldots, z_{nd}]\in\R^{nd\times d^\prime}$ stacks the directions in which $\ww$ is PPE, the j-th component of $\lambda_\tau$ is given by $\lambda_\tau^j=W_\tau^\top z_j$, and $|\tilde{W}_\tau|\leq (nd-d^\prime)\epsilon$ by \eqref{eq:claim}.
	
	\noindent
	II) \emph{Signal sequences are constrained by PPE}
	
	Pick any $k\in \N: n\leq k \leq T$ and $\epsilon >0$.
	Consider the signal $\QQ^k(\ww)$ in the interval $\tau =\tet, \ldots, \tet 
	+ T-k+1$, namely, $(w_{\tau}, \ldots, w_{\tau+k-1})\in \R^{kd}$. 
	For each subsequence of $k$ instants in the window $\tau =\tet, \ldots, T-k+1$, we can write $(k-n+1)$ $nd$-dimensional equations of the type \eqref{eq:PPE_implies_linear_dependence}. Compactly, they read
	\begin{equation}\label{eq:system}
		M 
		\underbrace{\begin{bmatrix}
				w_{\tau} \\
				\vdots\\
				w_{\tau+k-1}
		\end{bmatrix}}_{\coloneqq \hat{w}_{\tau}}
		=
		\tilde{E}
		\underbrace{\begin{bmatrix}
				\lambda_{\tau}\\
				\vdots\\
				\lambda_{\tau + k-n+1}
		\end{bmatrix}}_{\coloneqq \hat{\lambda}_{\tau}}
		+
		\underbrace{\begin{bmatrix}
				\tilde{W}_{\tau}\\
				\vdots\\
				\tilde{W}_{\tau + k-n+1}
		\end{bmatrix}}_{\coloneqq \tilde{w}_{\tau}},
	\end{equation}
	with $\hw_\tau \in \R^{kd}, \hat{\lambda}_\tau\in \R^{d^\prime(k-n+1)}$, $\tw_\tau \in \R^{nd(k-n+1)}$, and
	\begin{equation}
		\begin{split}
			M&\coloneqq 
			\begin{bmatrix}
				I_d & 0 & 0 & 0 & 0 & \ldots\\
				0 & \ddots & 0 & 0 & 0 & \ldots\\
				0 & 0 & \ddots & 0 & 0 & \ldots\\
				0 & 0 & 0 & I_d & 0 & \ldots\\
				0&I_d & 0 & 0 & 0& \ldots\\
				0&0 & \ddots & 0 & 0& \ldots\\
				0&0  &0& \ddots & 0 & \ldots\\
				0&0 & 0 & 0 & I_d  & \ldots\\
				\ldots&\ldots & \ldots & \ldots & \ldots & \ldots
			\end{bmatrix}\in \R^{nd(k-n+1)\times kd}\\
			\tilde{E}&\coloneqq I_{k-n+1} \otimes E \in \R^{nd(k-n+1)\times d^\prime(k-n+1)}.
		\end{split}
	\end{equation}
	We are interested in characterizing the solutions $(\hat{w}_{\tau}, \hat{\lambda}_{\tau})$ of \eqref{eq:system}, namely, the possible signals which fulfill the constraints which are imposed by PPE. In other words, we want to solve
	\begin{equation}
		\begin{bmatrix}
			M & -\tilde{E}
		\end{bmatrix}
		\begin{bmatrix}
			\hat{w}_{\tau} \\
			\hat{\lambda}_{\tau}
		\end{bmatrix}
		= \tilde{w}_{\tau}.
	\end{equation}
	
	\noindent
	III) \emph{A small enough degree of PPE overconstrains the signal sequence}
	
	Given the right pseudo-inverse 
	$\begin{bmatrix}
		M & -\tilde{E}
	\end{bmatrix}^\dagger$ of $\begin{bmatrix}
		M & -\tilde{E}
	\end{bmatrix}$, any solution $(\hat{w}_{\tau}, \hat{\lambda}_{\tau})$ can be written as
	\begin{equation}
		\begin{bmatrix}
			\hat{w}_{\tau} \\
			\hat{\lambda}_{\tau}
		\end{bmatrix}=
		\begin{bmatrix}
			M & -\tilde{E}
		\end{bmatrix}^\dagger \tilde{w}_\tau+ 
		v_\tau,
	\end{equation}
	where $v_\tau \in \ker\left(\begin{bmatrix} M & -\tilde{E}\end{bmatrix}\right)$. 
	Notice that, since $\begin{bmatrix} M & -\tilde{E}\end{bmatrix}$ has full row rank, by the rank-nullity theorem it holds that
	\begin{equation} 
		\begin{split}
			d^{\prime\prime}&\coloneqq 
			\dim(\ker\left(
			\begin{bmatrix} 
				M & -\tilde{E}
			\end{bmatrix}\right))\\
			&=\dim(\dom\left(\begin{bmatrix} 
				M & -\tilde{E}
			\end{bmatrix}\right))-\dim(\im\left(
			\begin{bmatrix} 
				M & -\tilde{E}
			\end{bmatrix}\right))\\
			&=
			kd + d^\prime(k-n+1) - nd(k-n+1)
			\\
			&=(k-n)(d+d^\prime-nd) + d^\prime,
		\end{split}
	\end{equation}
	from which it is clear that, if $d^\prime \leq d(n-1)$, then $d^{\prime\prime}\leq d^\prime$.
	Furthermore, notice that $|\tilde{w}_{\tau}|\leq \sqrt{nd(k-n+1)}\epsilon$ for all $\tau=\tet, \ldots, \tet+T-k+1$.
	This means that for all $\tau=\tet, \ldots, \tet+T-k+1$ we can write 
	\begin{equation}\label{eq:lower_dimensionality_dt}
		\hat{w}_{\tau} = 
		\begin{bmatrix}
			I_k & 0
		\end{bmatrix}
		G \nu_\tau + 
		\begin{bmatrix}
			I_k & 0
		\end{bmatrix}
		\begin{bmatrix}
			M & -\tilde{E}
		\end{bmatrix}^\dagger 
		\tilde{w}_\tau,
	\end{equation}
	with $G$ any surjective map
	\begin{equation}
		G: \R^{d^{\prime\prime}} \to \ker\left(
		\begin{bmatrix} 
			M & -\tilde{E}
		\end{bmatrix}\right),
	\end{equation}
	and $\nu_\tau \in \R^{d^{\prime\prime}}$ such that $v_\tau = G\nu_\tau$.
	Since $\nu_\tau\in \R^{d^{\prime\prime}}$, it can span in time at most $d^{\prime\prime}$ directions, and the same holds for $\begin{bmatrix}I_k & 0\end{bmatrix}G\nu_\tau$. Being $\tw_\tau$ arbitrarily small in the arbitrarily long interval $\tau = \tet, \ldots, \tet+T-k+1$, we conclude that $\QQ^k(\ww)=(q^{k-1}\ww, \ldots, q^0\ww)$ is PPE of degree at most $d^{\prime\prime}\leq d^\prime$.

	\subsection{Proof of Lemma \ref{lemma:CT_not_disagiato}}\label{proof:CT_not_disagiato}

	Pick any time interval $[t, t+T]$ of arbitrary length $T\geq 0$ and an unitary direction $z \in \R^d$. Being $\ww$ bounded, the quantity
	\begin{equation}
		\bw(t, T, z)\coloneqq \int_{t}^{t +T}|z^\top w(\tau)| \textup{d}\tau,
	\end{equation}  
	is finite for all $t, T>0$ and $z\in \R^d$. Choose an arbitrarily small $\alpha>0$, and define the sets
	\begin{equation}
		\begin{split}
			T_{>\alpha}(t, T, z)&\coloneqq \{\tau\in [t, t +T]: |z^\top w(\tau)|> \alpha \}\\
			T_{\leq\alpha}(t, T, z)&\coloneqq \{\tau\in [t, t +T]: |z^\top w(\tau)|\leq \alpha \}.
		\end{split}
	\end{equation}
	Notice that, by continuity of $\ww$, $T_{>\alpha}(t, T, z)$ is a union of open sets for any $t, T, z$. 
	Since $\|\ww\|_\infty\leq M$ for some $M>0$, then we can find an upper bound for the measure $\mu(T_{\geq\alpha}(t, T, z))$: 
	\begin{equation}
		\begin{split}
			\int_{t}^{t +T} |z^\top w(\tau)| \textup{d}\tau &= \bw(t, T, z) \\
			\int_{T_{> \alpha}} |z^\top w(\tau)| \textup{d}\tau + \int_{T_{\leq \alpha}} |z^\top w(\tau)| \textup{d}\tau  & = \bw(t, T, z)\\
			\int_{T_{> \alpha}} |z^\top w(\tau)| \textup{d}\tau & \leq \bw(t, T, z)
			\\
			\int_{T_{> \alpha}} \alpha \textup{d}\tau & \leq \bw(t, T, z)
			\\
			\alpha\mu(T_{\geq \alpha})&\leq \bw(t, T, z)
			\\
			\mu(T_{> \alpha})&\leq \alpha^{-1}\bw(t, T, z)
		\end{split}
	\end{equation}
	Overall, we obtained that for every $t, T, \alpha>0$, $z\in \R^d$ may partition the interval $[t, t+T]$ such that
	\begin{equation}
		\begin{split}
			|w(\tau)^\top z| &\leq \alpha, \;\;\; \forall \tau \in T_{\leq\alpha}(t, T, z),\\
			|w(\tau)^\top z| & > \alpha, \;\;\; \forall \tau \in T_{>\alpha}(t, T, z),
		\end{split}
	\end{equation}
	with $\mu(T_{>\alpha}(t, T, z))\leq \alpha^{-1}\bw(t, T, z)$.
	We want now to find an upper bound on $|w(\tau)^\top z|$ in the region $T_{>\alpha}(t, T, z)$. 
	
	\noindent
	Since $T_{>\alpha}(t, T, z)$ is a union of open sets of measure $\mu(T_{>\alpha})\leq \alpha^{-1}\bw(t, T, z)$, we consider the case in which it is a 
	unique interval (namely, the case in which $|w(\tau)|$ can grow more). Denote as $\bar{t}\in [t, t+T]$ the last time instant for which $|w(\bar{t})^\top z|\leq \alpha$, and $M\coloneqq \|\dd(\ww)\|_\infty$. 
	Then, we express any $\tau\in T_{> \alpha}(t, T, z)$ as $\tau=\bar{t}+\delta$, with $\delta \leq \alpha^{-1}\bw(t, T, z)$ (since we have shown that $\mu(T_{>\alpha})\leq \alpha^{-1}\bw(t, T, z)$). 
	It holds that
	\begin{equation}\label{eq:nonsohofinitoinomi}
		\begin{split}
			\left|w(\bar{t}+\delta)^\top z\right| &= \left|w(\bar{t})^\top z +\int_t^{t+\delta} \dw(\tau)^\top z\textup{d}\tau \right|\\
			|w(\bar{t}+\delta)^\top z| &\leq \alpha +M\delta\\
			|w(\bar{t}+\delta)^\top z|&\leq \alpha + M\alpha^{-1}\bw(t, T, z).
		\end{split}
	\end{equation}
	
	\noindent
	Choose any $\alpha>0$. Pick $\epsilon(\alpha)\coloneqq\frac{\alpha^2}{M}$. Then, for all $T>0$, since $\ww \notin \Cpe{d}$, there exists $\tet>0, z \in \R^d$ such that 
	\begin{equation}
		\begin{split}
			\bw(\tet, T, z) = \int_{t}^{t +T}|z^\top w(\tau)| \textup{d}\tau &\leq \epsilon.
		\end{split}
	\end{equation}
	Substituting in \eqref{eq:nonsohofinitoinomi}, we have
	\begin{equation}
		\begin{split}
			|w(\tau)^\top z| &\leq \alpha + M\alpha^{-1} \epsilon\\
			&\leq  \alpha + \alpha,
		\end{split}
	\end{equation}
	for all $\tau \in T_{>\alpha}(\tet, T, z)$, and thus for all $\tau \in [\tet, \tet+T]$.
	
	\noindent 
	To recap, we have proved that for any $\alpha, T>0$, there exists $\epsilon>0, z \in \R^d$ and thus $\tet>0$ for which
	$|w(\tau)^\top z|\leq 2\alpha$ for all $\tau \in [\tet, \tet + T]$, and this concludes the proof.

	\subsection{Proof of Lemma \ref{lemma:not_PE_in_CT}}\label{proof:not_PE_in_CT}
	
	We prove the theorem in three steps.
	
	\noindent
	I) \emph{A useful characterization for PPE signals.}
	
	Since $\WW(\ww)\coloneqq\DD^n(\ww)$ is PPE of degree at most $d^\prime \leq d(n-1)$, by Lemma \ref{lemma:CT_not_disagiato}, and following the same steps as in Lemma \ref{lemma:not_PE_in_DT} I), for all $T, \epsilon>0$ we can find $\tet>0$ such that
	\begin{equation}\label{eq:W_not_PE}
		W(\tau) = E \lambda(\tau) + \tilde{W}(\tau),
	\end{equation}
	where $E\in \R^{nd\times d^\prime}$ stacks the directions in which $\DD^n(\ww)$ is PPE, $\lambda(\tau)\in \R^{d^{\prime}}$ stacks the projections of $\DD^n(\ww)$ along these directions, and $|\tilde{W}(\tau)|\leq \epsilon$ for all $\tau\in [\tet, \tet+T]$.

	\noindent
	II) \emph{Sampled signal sequences are constrained by PPE.}
	
	Pick an arbitrary $\bar{T}\leq T$ and $N\in \N$. We define $\delta\coloneqq \bar{T}/N$.
	Notice that, being $\dd(\WW)$ Lipschitz continuous, it holds that
	\begin{equation}\label{eq:W_derivatives_0}
		W(\tau+\delta) = W(\tau) + \delta \dot{W}(\tau) + r(\delta),
	\end{equation}
	where $|r(\delta)|\leq c \frac{\delta^2}{2}$, $c\coloneqq \|\dd^2(\WW)\|_\infty$.
	
	\noindent
	By recalling that $W(\tau)=(w(\tau), \dw(\tau), \ldots, \dw^{(n-1)}(\tau))$, we can rewrite \eqref{eq:W_derivatives_0} as
	\begin{equation}\label{eq:W_derivatives}
		W(\tau+\delta)=(I+\delta S)W(\tau) + \delta e_d \dw^{(n)}(\tau)+ r(\delta),
	\end{equation}
	where $S\in \R^{nd\times nd}, e_d \in \R^{nd\times d}$ are given by
	\begin{equation}
		S\! =\! 
		\begin{bmatrix}
			0_{(n-1)d\times d} & I_{(n-1)d}\\
			0_{d\times d} & 0_{d\times (n-1)d}
		\end{bmatrix},\;\;
		e_d \!=\! 
		\begin{bmatrix}
			0_{d\times d}  & 0_{d\times d}  & \ldots & I_d
		\end{bmatrix}^{\top}.
	\end{equation}
	By substituting \eqref{eq:W_not_PE} into \eqref{eq:W_derivatives}, we obtain that in the interval $\tau \in [\tet, \tet+T-\delta]$ we can write
	\begin{equation}
		\begin{split}
			E\lambda&(\tau+\delta) + \tilde{W}(\tau + \delta) =\\
			&\!= \!(I_{nd}+\delta S)(E\lambda(\tau) + \tilde{W}(\tau))
			+ \delta e_d \dw^{(n)}(\tau)+ r(\delta).\;\;\;\;\;\\
		\end{split}
	\end{equation}

	Considering $N$ consecutive instants $\tau, \tau+\delta, \ldots, \tau \!+\! \bar{T}\!-\!\delta$, similarly to \eqref{eq:system}  we obtain 
	\begin{equation}\label{eq:system_ct}
		M \hw(\tau) = \tilde{E} \hat{\lambda}(\tau) + \tw(\tau),
	\end{equation}
	where $\hw(\tau)\in \R^{Nd}, \hat{\lambda}(\tau)\in\R^{d^\prime(N+1)}, \tw(\tau)\in \R^{ndN}$ are given by
	\begin{equation}\label{eq:various_vectors}
		\begin{split}
			\hw(\tau) &\coloneqq
			\begin{bmatrix}
				\dw^{(n)}(\tau)\\
				\dw^{(n)}(\tau + \delta)\\
				\vdots\\
				\dw^{(n)}(\tau + \bar{T}-\delta)
			\end{bmatrix},
			\hat{\lambda}(\tau)\coloneqq
			\begin{bmatrix}
				\lambda(\tau)\\
				\lambda(\tau + \delta)\\
				\vdots\\
				\lambda(\tau + \bar{T})
			\end{bmatrix},\\
			\tw(\tau)&\coloneqq
			\begin{bmatrix}
				-(I_{nd}+\delta S)\tilde{W}(\tau) + \tilde{W}(\tau+\delta) + r(\delta)\\
				-(I_{nd}+\delta S)\tilde{W}(\tau+\delta) + \tilde{W}(\tau+2\delta)  + r(\delta)\\
				\vdots\\
				-(I_{nd}+\delta S)\tilde{W}(\tau \!+\! \bar{T}-\delta)  +  \tilde{W}(\tau + \bar{T})+ r(\delta)
			\end{bmatrix},
		\end{split}
	\end{equation}
	and
	\begin{equation}
		\begin{split}
			M& \coloneqq \delta I_{N}\otimes e_d \in \R^{ndN\times Nd}\\
			\tilde{E}&\coloneqq 
			\begin{bmatrix}
				[-(I_{nd}+\delta S)E\;\;  E]\;\;\;0\;\;\;\ldots\;\;\; 0 \\
				0 \;\;\; [-(I_{nd}+\delta S)E\;\;  E] \;\;\; 0 \;\;\;\ldots\\
				\vdots\\
				0\;\;\;\ldots \;\;\;0\;\;\;[-(I_{nd}+\delta S)E\;\;  E]  \\
			\end{bmatrix}\in \R^{ndN\times d^\prime(N+1)}.
		\end{split}
	\end{equation}
	We are interested in characterizing the solutions $(\hw(\tau), \hat{\lambda}(\tau))$ of system \eqref{eq:system_ct}, namely, the possible signals which fulfill the constraints which are imposed by PPE. In other words, we want to solve
	\begin{equation}
		\begin{bmatrix}
			M & -\tilde{E}
		\end{bmatrix}
		\begin{bmatrix}
			\hw(\tau) \\
			\hat{\lambda}(\tau)
		\end{bmatrix}
		= \tw(\tau)
	\end{equation}
	in the interval $[\tet, \tet+T]$.

	\noindent
	III) \emph{A small enough degree of PPE overconstrains the sampled signal sequence.}
	
	Given the right pseudo-inverse 
	$\begin{bmatrix}
		M & -\tilde{E}
	\end{bmatrix}^\dagger$ of $\begin{bmatrix}
		M & -\tilde{E}
	\end{bmatrix}$, any solution $(\hw(\tau), \hat{\lambda}(\tau))$ can be written as
	\begin{equation}\label{eq:lambda_sol}
		\begin{bmatrix}
			\hw(\tau) \\
			\hat{\lambda}(\tau)
		\end{bmatrix}=
		\begin{bmatrix}
			M & -\tilde{E}
		\end{bmatrix}^\dagger \tw(\tau)+ 
		v (\tau),
	\end{equation}
	where $v(\tau) \in \ker\left(\begin{bmatrix} M & -\tilde{E}\end{bmatrix}\right)$. 
	Consider the term $\tw(\tau)$: we show at first it can be made arbitrarily small. By recalling the bounds on $r(\delta)$ (see \eqref{eq:W_derivatives_0}) and $\tilde{W}(\tau)$ (see \eqref{eq:W_not_PE}), it holds that
	\begin{equation}\label{eq:bound_for_tw}
		\begin{split}
			|\tw(\tau)|&\leq \left|
			\begin{bmatrix}
				(I_{nd}+\delta S)\tilde{W}(\tau) \\
				(I_{nd}+\delta S)\tilde{W}(\tau+\delta)\\
				\vdots\\
				(I_{nd}+\delta S)\tilde{W}(\tau + \bar{T}-\delta)
			\end{bmatrix}\right| 
			+
			\left|
			\begin{bmatrix}
				\tilde{W}(\tau+\delta) \\
				\tilde{W}(\tau+2\delta)\\
				\vdots\\
				\tilde{W}(\tau +\bar{T})
			\end{bmatrix}\right|
			+ \left|
			\begin{bmatrix}
				r(\delta)\\
				r(\delta)\\
				\vdots\\
				r(\delta)
			\end{bmatrix}\!
			\right|
			\\
			&\leq (1+\delta)\sqrt{N}\epsilon + \sqrt{N}\epsilon + \sqrt{N}c\frac{\delta^2}{2}
			\\
			& \leq \left(2\sqrt{N}  + \frac{\bar{T}}{\sqrt{N}}\right)\epsilon + \sqrt{N} c\frac{\bar{T}^2}{2N^2}.
		\end{split}
	\end{equation}
	We now show we can always find an arbitrarily long interval in which $\tw(\tau)$ is arbitrarily small. 
	Notice that, given any $\epsilon^\prime>0$, by choosing
	\begin{equation}
		N\geq \bar{N}\coloneqq \sqrt[3]{\frac{\bar{T}^4 c^2}{\epsilon^{\prime2}}} \implies \sqrt{N} c\frac{\bar{T}^2}{2N^2}\leq \frac{\epsilon^\prime}{2}.
	\end{equation}
	Exploting this, and choosing $\epsilon=\frac{\epsilon^\prime}{2(2\sqrt{N} + \bar{T}/\sqrt{N})}$, we have from \eqref{eq:bound_for_tw} that for any $\epsilon^\prime>0, \bar{T}>0, T>\bar{T}$ there exists $\bar{N}\in\N$, $\tet>0$ such that, for all $N\geq \bar{N}$,
	\begin{equation}\label{eq:tw_bound}
		|\tw(\tau)| \leq \epsilon^\prime \;\;\;\forall \tau \in [\tet, \tet+T].
	\end{equation}

	\noindent 
	We move now to the term $v(\tau)\in \ker\left(\begin{bmatrix} M & -\tilde{E}\end{bmatrix}\right)$.
	Notice that, since $\begin{bmatrix} M & -\tilde{E}\end{bmatrix}$ has full row rank, by the rank-nullity theorem it holds
	\begin{equation} 
		\begin{split}
			d^{\prime\prime}&\coloneqq 
			\dim(\ker\left(
			\begin{bmatrix} 
				M & -\tilde{E}
			\end{bmatrix}\right))\\
			&=\dim(\dom\left(\begin{bmatrix} 
				M & -\tilde{E}
			\end{bmatrix}\right))-\dim(\im\left(
			\begin{bmatrix} 
				M & -\tilde{E}
			\end{bmatrix}\right))\\
			&=
			Nd + d^\prime(N+1)-ndN
			\\
			&=N(d+d^\prime-nd) + d^\prime,
		\end{split}
	\end{equation}
	from which it is clear that if $d^\prime \leq d(n-1)$ then $d^{\prime\prime}\leq d^\prime$.
	This means that we can write the solution $\hat{\lambda}(\tau)$ of \eqref{eq:lambda_sol} as
	\begin{equation}\label{eq:lower_dimensionality_ct}
		\hat{\lambda}(\tau) = \begin{bmatrix}
			0 & I_N
		\end{bmatrix}
		G  \nu(\tau) + 
		\begin{bmatrix}
			0 & I_N
		\end{bmatrix}
		\begin{bmatrix}
			M & -\tilde{E}
		\end{bmatrix}^\dagger \tw(\tau),
	\end{equation}
	with $G$ any surjective map
	\begin{equation}
		G: \R^{d^{\prime\prime}} \to \ker\left(
		\begin{bmatrix} 
			M & -\tilde{E}
		\end{bmatrix}\right),
	\end{equation}
	and $\nu(\tau) \in \R^{d^{\prime\prime}}$ such that $v(\tau) = G\nu(\tau)$.
	Using \eqref{eq:W_not_PE} and \eqref{eq:lower_dimensionality_ct}, and recalling the definition of $\hat{\lambda}(\tau)$ in \eqref{eq:various_vectors}, we can reconstruct the vector in which we are interested in:
	\begin{equation}
		\begin{split}
			\begin{bmatrix}
				W(\tau)\\
				W(\tau + \frac{\bar{T}}{N})\\
				\vdots\\
				W(\tau + \bar{T})
			\end{bmatrix}=&
			G^\prime \nu(\tau) + F \tw(\tau) +
			\begin{bmatrix}
				\tilde{W}(\tau)\\
				\tilde{W}(\tau + \frac{\bar{T}}{N})\\
				\vdots\\
				\tilde{W}(\tau + \bar{T})
			\end{bmatrix},
		\end{split}
	\end{equation}
	where 
	\begin{equation}
		\begin{split}
			G^\prime &= (I_N \otimes E) 
			\begin{bmatrix}
				0 & I_N
			\end{bmatrix} G\\
			F &= (I_N \otimes E) 
			\begin{bmatrix}
				0 & I_N
			\end{bmatrix}
			\begin{bmatrix}
				M & -E
			\end{bmatrix}^\dagger.
		\end{split}
	\end{equation}
	
	The proof is complete by recalling that the choices of $N, \epsilon$ are such that $|\tw(\tau)|\leq \epsilon^\prime$ (in \eqref{eq:tw_bound}) and $|(\tilde{W}(\tau), \ldots, \tilde{W}(\tau+\bar{T}))|\leq \epsilon^\prime/2$ for all $\tau \in [\tet, \tet+T]$.

	\subsection{Proof of Lemma \ref{lemma:PE_of_derivatives}}\label{proof:PE_of_derivatives}
	
	By Lemma \ref{lemma:CT_not_disagiato}, if $\ww\notin \Cpe{d}$ then for all $T, \epsilon >0$ we can find $\tet>0$, $z\in \R^d$ such that
	\begin{equation}
		|w(\tau)^\top z|\leq \epsilon, \;\;\; \forall \tau\in [\tet, \tet+T].
	\end{equation}
	Picking any $\tau, \delta>0$ such that $\tau, \tau + \delta \in [\tet, \tet+T]$, we have that
	\begin{equation}\label{eq:signal_differences}
		|(w(\tau+\delta)- w(\tau))^\top  z| \leq |w(\tau+\delta)^\top z| +  |w(\tau)^\top z| \leq 2\epsilon
	\end{equation}
	Expanding $w(\tau+\delta)$ in Taylor series, we obtain that
	\begin{equation}\label{eq:taylor_on_signal}
		w(\tau+\delta)-w(\tau) = \dw(\tau) \delta  + o(\delta),
	\end{equation}
	where $|o(\delta)|\leq M\frac{\delta^2}{2}$, with $M\coloneqq\|\dd^2(\ww)\|_\infty$.
	Using \eqref{eq:taylor_on_signal} in \eqref{eq:signal_differences}, we obtain
	\begin{equation}
		\begin{split}
			|(w(\tau+\delta)- w(\tau))^\top z| &\leq 2\epsilon\\
			|(\dw(\tau) \delta  + o(\delta))^\top z| &\leq 2\epsilon.
		\end{split}
	\end{equation}
	Since, by the triangle inequality, $|(\dw(\tau) \delta  + o(\delta))^\top z|\geq |\dw(\tau)^{\top}z \delta |- |o(\delta)^\top z|$, we have that
	\begin{equation}
		\begin{split}
			|\delta \dw(\tau)^\top z| - |o(\delta)^\top z| &\leq 2\epsilon\\
			|\delta \dw(\tau)^\top z | &\leq 2\epsilon + |o(\delta)|\\
			|\dw(\tau)^\top z| & \leq 2\frac{\epsilon}{\delta} + M\frac{\delta}{2}.
		\end{split}
	\end{equation}

	Choosing a sufficiently small $\epsilon>0$ and $\delta(\epsilon)\coloneqq \sqrt{\epsilon}$, for all $\tau\in [\tet, \tet+T-\delta(\epsilon)]$ we have
	\begin{equation}
		|\dw(\tau)^\top z| \leq \left(2+\frac{M}{2}\right)\sqrt{\epsilon}.
	\end{equation}
	By defining $\gamma(\epsilon) \coloneqq \max(\sqrt{\epsilon},(2+\frac{M}{2})\sqrt{\epsilon})$, we obtain
	\begin{align}\label{eq:bound_on_both_w_and_dw}
		|w(\tau)^\top z|\leq \gamma(\epsilon) &&& 	|\dw(\tau)^\top z| \leq \gamma(\epsilon)
	\end{align}
	for all $\tau\in [\tet, \tet+T-\delta(\epsilon)]$. 
	
	\noindent 
	Since $\gamma(\epsilon)$ is a strictly increasing function of $\epsilon$ such that $\gamma(0)=0$, we may pick an arbitrarily small $\epsilon^\prime$ and find $\epsilon: \epsilon^\prime=\gamma(\epsilon)$;
	consequently, since $\tet$ exists for any choice of $\epsilon, T>0$, for any choice of $\epsilon^\prime$ there exists an interval for which \eqref{eq:bound_on_both_w_and_dw} holds.

	\subsection{Proof of Theorem \ref{thm:nec_mi_dt}}\label{proof:nec_mi_dt}
	
	We prove this result by contraposition, i.e., we show that given $\Sigma \in \sctrbx$, if $\QQ^n(\uu)\in \D{nm}$ is PPE of degree at most $n^\prime \leq n-1$, then for all $x_0\in \R^n$, $\xx=\Sigma(\uu, x_0)\notin \Dpe{n}$ regardless of the initial condition.
	
	\noindent
	I) \emph{$x_\tau$ can be approximated arbitrarily well by a linear function of past inputs.}
	
	Let $x_0=0$. 
	We can write system dynamics as
	\begin{equation}\label{eq:dynamic_for_subsystem_dt_3}
		\begin{split}
			x_{\tau+1}  &= A^{n}x_{\tau-n+1} + RU_\tau,
		\end{split}
	\end{equation}
	where $R$ is the reachability matrix and $U_\tau=(u_{\tau}, \ldots, u_{\tau-n+1})\in \R^{nm}$. By writing $x_{\tau-n+1}$ as a function of the previous inputs, and repeating this recursion for arbitrary $K$ steps, we obtain
	\begin{equation}\label{eq:x_as_fcn_of_U_1_mi}
		\begin{split}
			x_{\tau+1} =& A^{Kn}x_{\tau-Kn+1}\\
			&\!+\!
			\underbrace{
				\begin{bmatrix}
					I & A^n & \ldots & A^{Kn}
				\end{bmatrix}
			}_{\eqqcolon \bar{A}}\!
			\underbrace{
				(I_K \otimes R)
			}_{\eqqcolon \bar{R}}\!
			\begin{bmatrix}
				U_\tau \\
				U_{\tau-n}\\
				\vdots\\
				U_{\tau-(K-1)n}
			\end{bmatrix}.\!\!\!
		\end{split}
	\end{equation}
	Notice that for any $\epsilon>0$, we can choose $K>0$ such that $|A^{Kn} x_{\tau-Kn+1}|\leq \epsilon$ for all $\tau\in \N$, being the signal $\xx$ bounded.
	Notice the vector $(U_\tau, \ldots, U_{\tau-(K-1)n})$ is given by the signal $\QQ^{(K-1)n}(\uu)$ evaluated at time $\tau$.

	\noindent
	II) \emph{The degree of PPE of $\QQ^{(K-1)n}(\uu)$ is limited by the degree of PPE of $\QQ^n(\uu)$.}
	
	If $\QQ^n(\uu)$ is PPE of degree at most $n^\prime \leq n-1$, then $n^\prime \leq m(n-1)$ and we can apply Lemma \ref{lemma:not_PE_in_DT} to ensure that for any $K\geq 1$, $\QQ^{(K-1)n}(\uu)$ is PPE of degree at most $n^\prime \leq n-1$.
	
	\noindent
	III) \emph{If the spanned directions are $n-1$, $\xx\notin \Cpe{n}$.}
	
	Leveraging the lack of PPE demonstrated in the above point, for all $T, K, \epsilon >0$ there exists $\tet>0$ for which we can rewrite \eqref{eq:x_as_fcn_of_U_1_mi} as (see the derivation in Lemma \ref{lemma:not_PE_in_DT})
	\begin{equation}\label{eq:x_as_fcn_of_U_2_mi}
		x_\tau = A^{Kn}x_{\tau-Kn} + \bar{A} \bar{R} E \lambda_t + \bar{A} \bar{R}\tilde{\lambda}_t, 
	\end{equation}
	for all $\tau = \tet, \ldots, \tet + T$, for some $E \in \R^{(K-1)nm\times (n-1)}$ stacking the directions in which there is PPE, $\lambda_t \in \R^{n-1}$ stacking the projections of $(U_t, \ldots, U_{t-(K-1)n})$ along these directions, $\tilde{\lambda}_t\in \R^{(K-1)nm}$ such that $|\tilde{\lambda}_{t_{\epsilon, K}}|\leq \epsilon$. 
	
	Since $\bar{A}\bar{R}E \in \R^{n\times (n-1)}$, there exists $z \in \R^{n}$ such that $z^\top \bar{A}\bar{R}E=0$, which implies that
	\begin{equation}
		z^\top x_\tau = z^\top A^{Kn}x_{\tau-Kn} +z^\top \bar{A} \bar{R}\tilde{\lambda}_t.
	\end{equation}
	Being both $|A^{Kn}x_{\tau-Kn}|$ and $\tilde{\lambda}$ arbitrarily small in the arbitrarily long interval $\tau=\tet+n-1, \tet+T$, we conclude $\xx=\Sigma(\uu, 0)\notin \Dpe{n}$ and thus $\uu \notin \Csr(\Sigma, 0)$. To conclude the proof, since $\Sigma$ is a stable linear system, we have that all solutions $\xx=\Sigma(\uu, x_0)$ converge exponentially to those initialized in $x_0=0$. Since a vanishing term cannot guarantee PE of a signal, we can conclude that for all $x_0\in \R^n$, if $\QQ^{n}(\uu)$ is PPE of degree at most $n^\prime \leq n-1$, then $\xx=\Sigma(\uu, x_0) \notin \Dpe{n}$.

	We pass now to the second statement of the theorem. Since it has an analogous proof, we sketch only the main differences with respect to the proof given before. 
	We want to prove that, given $\Sigma \in \sctrbxu$, if $\QQ^{n+1}(\uu)\in \D{nm}$ is PPE of degree $n^\prime \leq n+m-1$, then for all $x_0\in \R^n$, $\xx=\Sigma(\uu, x_0)\notin \Dpe{n}$ regardless of the intitial condition.
	
	\noindent
	I) \emph{$(x_\tau, u_\tau)$ can be approximated arbitrarily well by a linear function of the previous inputs.}
	
	Similarly as done in \eqref{eq:x_as_fcn_of_U_1_mi}, we obtain
	\begin{equation}\label{eq:xu_as_fcn_of_U_mi}
		\begin{bmatrix}
			u_{\tau+1}\\
			x_{\tau+1}
		\end{bmatrix}
		=\begin{bmatrix}
			0 \\
			A^{Kn}x_{\tau-Kn+1}
		\end{bmatrix}
		+
		\begin{bmatrix}
			I_m & 0\\
			0 & \bar{A} \bar{R}
		\end{bmatrix}
		\begin{bmatrix}
			u_{\tau+1}\\
			U_\tau\\
			\vdots\\
			U_{\tau-(K-1)n}\\
		\end{bmatrix}\!\!,
	\end{equation}
	where $\bar{A}, \bar{R}$ are the same as in \eqref{eq:xu_as_fcn_of_U_mi}, and by choosing an appropriate $K\in \N$, for any $\epsilon>0$ we achieve $|A^{Kn}x_{\tau-Kn+1}|\leq \epsilon$, since $A$ is Shur and $\xx$ is bounded.
	Notice that the vector $(u_{\tau+1}, U_\tau, \ldots, U_{\tau-(K-1)n})$ is given  by the signal $\QQ^{(K-1)n+1}(\uu)$ evaluated at time $\tau+1$.

	\noindent
	II) \emph{The degree of PPE of $\QQ^{(K-1)n+1}(\uu)$ is limited by the degree of PPE of $\QQ^{n+1}(\uu)$.}
	
	Applying Lemma \ref{lemma:not_PE_in_DT}, if $\QQ^{n+1}(\uu)$ is PPE of degree at most $n^\prime \leq n+m-1$, then, since $n+m-1 \leq m(n+1-1)=nm$ for all $n, m\in \N$, $\QQ^{(K-1)n+1}(\uu)$ is PPE of degree at most $n^\prime$.
	
	\noindent
	III) \emph{If the spanned directions are $n\!+\!m\!-\!1$, then $(\xx, \uu)\notin \Dpe{n+m}$.}
	
	This point proceeds exactly as for the previous case, so we omit it.

	\subsection{Proof of Theorem \ref{thm:suff_mi_dt}}\label{proof:suff_mi_dt}
	
	We prove that $\QQ^n(\uu)\in \Dpe{nm} \implies \Sigma(\uu, x_0)\in \Dpe{n}$ for $\Sigma \in \sctrbx$ and for all $x_0\in \R^n$ by contraposition, i.e., we show that for all $x_0\in \R^n$, $\Sigma(\uu, x_0)\notin \Dpe{n} \implies \QQ^n(\uu)\notin \Dpe{nm}$.
	We prove this in four points.
	
	\noindent
	I) \emph{The lack of PE of $\xx$ constrains the system input.}
	
	If $\xx\coloneqq\Sigma(\uu, x_0)\notin \Dpe{n}$, applying Definition \ref{def:PE_DT} we have that for all $T, \epsilon>0$ we can find a direction $z\in \R^n$ and $\tet\in\N$ such that
	\begin{equation}\label{eq:small_x}
		|z^\top x_\tau|\leq \epsilon
	\end{equation}
	for all $\tau =\tet,\ldots, \tet+T$.
	Along direction $z$, system dynamics read
	\begin{equation}
		z^\top x_{\tau+1} = z^\top Ax_\tau + z^\top Bu_\tau.
	\end{equation}
	If $z^\top B \neq 0$, we can find the input $u_\tau$ as
	\begin{equation}\label{eq:closed_loop_u_mi}
		u_\tau = \underbrace{-(z^\top B)^\dagger z^\top A}_{\coloneqq K}x_\tau + \underbrace{(z^\top B)^\dagger z^\top x_{\tau+1}}_{\coloneqq \tu_\tau} + v_\tau,
	\end{equation}
	where $|\tu_\tau|\leq |(z^\top B)^\dagger|\epsilon$, and $v_\tau \in \ker(z^\top B)$.
	
	\noindent
	II) \emph{The closed-loop dynamics depend only on $x_\tau, v_\tau$.}
	
	The dynamics become
	\begin{equation}\label{eq:closed_loop_dynamics_mi}
		\begin{split}
			u_\tau &=Kx_\tau +  \tu_\tau + v_\tau,\\
			x_{\tau+1}&=\underbrace{(A - BK)}_{\coloneqq \tilde{A}}x_\tau + B\tu_\tau + Bv_\tau
		\end{split}
	\end{equation}		
	where, for all $\tau = \tet, \ldots, \tet+T-1$, it holds that $|B\tu_\tau|\leq \epsilon$ and $v_\tau \in \ker(z^\top B)$.
	Consider the signal $\UU(\uu)\coloneqq\QQ^n(\uu)$. Given \eqref{eq:closed_loop_u_mi}, we obtain
	\begin{equation}\label{eq:U_expression}
		U_\tau =
		(I_n \otimes K )
		\begin{bmatrix}
			x_{\tau-n+1}\\
			\vdots\\
			x_{\tau}
		\end{bmatrix}
		+
		\underbrace{\begin{bmatrix}
				\tu_{\tau-n+1}\\
				\vdots\\
				\tu_{\tau}
		\end{bmatrix}}_{\coloneqq \tU_\tau}
		+
		\underbrace{\begin{bmatrix}
				v_{\tau-n+1}\\
				\vdots\\
				v_{\tau}
		\end{bmatrix}}_{\coloneqq V_\tau}.
	\end{equation}
	Using \eqref{eq:closed_loop_dynamics_mi}, it holds that
	\begin{equation}\label{eq:X_expr}
		\begin{bmatrix}
			x_{\tau-n+1}\\
			\vdots\\
			x_{\tau}
		\end{bmatrix}\!=\!
		Fx_{\tau-n+1}
		+ G\left(
		\tU_\tau + V_\tau
		\right),
	\end{equation}
	where 
	\begin{equation}
		\small
		F=
		\begin{bmatrix}
			I\\
			\vdots\\
			\tilde{A}^{n-1}
		\end{bmatrix},\;\;\;
		G = 
		\begin{bmatrix}
			0 & \ldots & \ldots & 0\\
			B & 0 & \ldots & 0\\
			\vdots& \vdots& \vdots & \vdots\\
			A^{n-2}B & \ldots & B & 0
		\end{bmatrix},
	\end{equation}
	and substituting \eqref{eq:X_expr} in \eqref{eq:U_expression}, we obtain
	\begin{equation}
		\begin{split}
			U_\tau &= \underbrace{(I_n \otimes K)F}_{\coloneqq F^\prime}x_{\tau-n+1} + \underbrace{((I_n \otimes K)G+I)}_{G^\prime}
			\left(
			\tU_\tau\!+\!\!
			{V}_\tau
			\right),\\
			&=
			\begin{bmatrix}
				F^\prime & G^\prime
			\end{bmatrix}
			\begin{bmatrix}
				x_{\tau-n+1} \\ 
				V_\tau
			\end{bmatrix}
			+G^\prime\tU_\tau.
		\end{split}
	\end{equation}

	\noindent
	III) \emph{Lack of PE in $\xx$ implies $\QQ^{n}(\uu)$ is not PE.} 
	
	Notice that in the period $\tau=\tet+n-1, \ldots, \tet+T$, it holds that
	\begin{enumerate}
		\item[i)] $v_\tau \in \ker(z^\top B)$, with $\dim(\ker(z^\top B))\leq m-1$, by construction of $v_\tau$.
		\item[ii)] $|z^\top x_\tau|\leq \epsilon$ by lack of PE of $\xx$. 
		\item[iii)] $|\tU_\tau|\leq \sqrt{n}|(z^\top B)^\dagger|\epsilon$ from \eqref{eq:closed_loop_u_mi}.
	\end{enumerate}
	Holding i) and ii), the space persistently spanned by $(x_\tau, V_\tau)$ is at most $(n-1)+n(m-1)=nm-1$ dimensional. In other words, by following the same procedure as in Lemma \ref{lemma:not_PE_in_DT}, we can write for all $\tau=\tet+n-1, \tet+T$
	\begin{equation}
		\begin{bmatrix}
			x_{\tau-n+1} \\ 
			V_\tau
		\end{bmatrix}=
		E \lambda_\tau + \tilde{\lambda}_\tau,
	\end{equation}
	where $E \in \R^{(nm)\times (nm-1)}$ stacks the directions which are persistently spanned by $(x_\tau, V_\tau)$, $\lambda_\tau \in \R^{nm-1}$ stacks the projections of $(x_\tau, V_\tau)$ along these directions, and $\tilde{\lambda}_\tau\in \R^{nm}, |\tilde{\lambda}_\tau|\leq \epsilon$ is an arbitrarily small perturbation.
	Finally, we have
	\begin{equation}
		U_\tau = 
		\underbrace{\begin{bmatrix}
				F^\prime & G^\prime
			\end{bmatrix} E}_{\coloneqq H}\lambda_\tau + 
		\begin{bmatrix}
			F^\prime & G^\prime
		\end{bmatrix} \tilde{\lambda}_\tau
		+G^\prime\tU_\tau,
	\end{equation}
	and since $\begin{bmatrix}
		F^\prime & G^\prime
	\end{bmatrix} E \in \R^{nm\times (nm-1)}$, there exists $z \in \R^{nm}$ such that $z^\top \begin{bmatrix}
		F^\prime & G^\prime
	\end{bmatrix} E=0$, which implies
	\begin{equation}
		z^\top U_\tau = z^\top
		\begin{bmatrix}
			F^\prime & G^\prime
		\end{bmatrix} \tilde{\lambda}_\tau
		+z^\top G^\prime\tU_\tau.
	\end{equation}
	Being both $\tU_\tau$ and $\tilde{\lambda}$ arbitrarily small in the arbitrarily long interval $\tau=\tet+n-1, \ldots, \tet+T$, we conclude $\QQ^n(\uu)\notin \Dpe{nm}$.
	
	\noindent
	IV) \emph{The case of $z^\top B=0$.}
	
	Consider the case where $z^\top B = 0$, namely, the columns $b_1, \ldots, b_m$ of $B$ satisfy $b_1, \ldots, b_m\in \ker(z^\top)$. In that case, along direction $z$, the system dynamics read as
	\begin{equation}\label{eq:constrained_dynamics_mi}
		\begin{split}
			z^\top x_{\tau+1} &= z^\top Ax_\tau + z^\top Bu_\tau\\
			&= z^\top Ax_\tau.
		\end{split}
	\end{equation}
	We can distinguish two cases: either $\ker(z^\top)=\ker(z^\top A)$ or $\ker(z^\top)\neq \ker(z^\top A)$.
	If $\ker(z^\top)=\ker(z^\top A)$, then $\ker(z^\top)$ must be an invariant subspace of $A$ of dimension $n-1$.
	Since $\ker(z^\top)$ is $A$-invariant, and since $b_1, \ldots, b_m\in \ker(z^\top)$, there are at most $n-1$ linearly independent vectors between the columns of $B, \ldots, A^{n-1}B$. This is a contradiction, since we assumed $(A, B)$ controllable, thus, it cannot hold $\ker(z^\top)=\ker(z^\top A)$.
	
	\noindent
	At last, consider the case $\ker(z^\top)\neq \ker(z^\top A)$. 
	Since \eqref{eq:constrained_dynamics_mi} holds for all $\tau = \tet+1, \dots, \tet+T-1$, we can write 
	\begin{equation}
		\begin{split}
			z^\top x_\tau &\leq \epsilon\\
			z^\top x_{\tau + 1}  = z^\top Ax_\tau & \leq \epsilon,
		\end{split}
	\end{equation}
	from which we deduce that we can write $x_\tau = \bx_{\tau}+ \tx_\tau$, with $\bx_\tau \in \ker(z^\top)\cap \ker(z^\top A)$ and $|\tx_{\tau}|\leq \epsilon$. 
	Since $\ker(z^\top)\neq \ker(z^\top A)$ and since the dimension of each kernel is at most $n-1$, then $\dim(\ker(z^\top)\cap \ker(z^\top A))\leq n-2$.
	
	\noindent
	This means that $\bx_\tau$ spans persistently at most $n-2$ directions, namely, there exists an unitary $z_2\in \R^{n}, z_2\perp z \perp (\ker(z^\top)\cap \ker(z^\top A))$, such that 
	\begin{equation}
		|y_2^\top x_\tau | = |y_2^\top \bx_\tau + y_2^\top \tx_\tau| =|y_2^\top \tx_\tau| \leq \epsilon.
	\end{equation}
	We can thus repeat the same procedure as before, checking if $z_2^\top B \neq 0$. If $z_2^\top B \neq 0$, we can repeat the reasoning in points I), II), III); otherwise, we can repeat the above reasoning to find another direction $z_3$ such that $z_3^\top x_\tau \leq \epsilon$ (and repeat the process until we find some $z_i^\top B\neq0$, which must exist since $B \neq 0$.
	In each of these cases, $\QQ^n(\uu)\notin \Dpe{nm}$, which concludes the proof for the first statement of the theorem.
	
	We move to the second statement of the theorem. Since it has an analogous proof, we sketch only the main differences with respect to the proof given before. We want to prove by contraposition that $\QQ^n(\uu)\in \Dpe{(n+1)m} \implies \Sigma(\uu, x_0)\in\Dpe{n+m}$ for $\Sigma \in \sctrbxu$ and for all $x_0\in\R^n$ by showing that for all $x_0\in \R^n$, $\Sigma(\uu, x_0)\notin \Dpe{n+m}\implies \QQ^{n+1}(\uu)\notin \Dpe{(n+1)m}$.

	\noindent
	I) \emph{The lack of PE of $(\xx, \uu)$ constrains the system input.}
	
	If $(\xx, \uu)\notin \Dpe{n+m}$ we have that for all $T, \epsilon>0$ we can find $z=(z_x, z_u)\in \R^{n+m}$, $\tet \in \N$ such that
	\begin{equation}
		z_x^\top x_\tau = -z_u^\top u_\tau + \chi_\tau,
	\end{equation}
	where $\chi_\tau \in \R, |\chi_\tau|\leq \epsilon$ for all $\tau = \tet, \ldots, \tet+T$. By pre-multiplying by $z_x^\top$ the systems dynamics, we obtain the update
	\begin{equation}\label{eq:next_u}
		u_{\tau+1} = -(z_u^\top)^\dagger z_x^\top \left( Ax_\tau + Bu_\tau \right) +(z_u^\top)^\dagger \chi_{\tau+1} + v_{\tau+1},
	\end{equation}
	where $v_{\tau+1}\in \ker(z_u^\top)$.

	\noindent
	II) \emph{The closed-loop dynamics depends only on $x_\tau, v_\tau, u_\tau$.}
	
	Since we can write $x_{\tau+1}=Ax_\tau + Bu_\tau$, in the interval $\tau=\tet, \ldots, \tet+T$ we can use \eqref{eq:next_u} to express each $u_\tau, u_{\tau+1}, \ldots, u_{\tau + n}$ as a linear function of $x_\tau, u_\tau, v_{\tau +1}, \ldots, v_{\tau + n}$ (similarly as done in \eqref{eq:U_expression}) plus an arbitrarily small quantity $\tU_\tau$. In other words, we have
	\begin{equation}\label{eq:U_not_spanning}
		\begin{bmatrix}
			u_\tau\\
			\vdots\\
			u_{\tau + n}
		\end{bmatrix}
		=
		K
		\begin{bmatrix}
			u_\tau\\
			x_\tau\\
			v_{\tau + 1}\\
			\vdots\\
			v_{\tau + n}
		\end{bmatrix}
		+ \tU_\tau.
	\end{equation}

	\noindent
	III) \emph{A lack of PE in $(\xx, \uu)$ implies $\QQ^{n+1}(\uu)$ is not PE.}
	
	Since $(\xx, \uu)$ is not PE, it spans at most $n+m-1$ directions. Since $v_\tau \in \ker(z_u^\top)$, the vector $(v_{\tau+1}, \ldots, v_{\tau+n})$ spans at most $(m-1)n$ directions. Overall, we have that the right-hand side of \eqref{eq:U_not_spanning} spans persistently only $n+m-1 + (m-1)n= (n+1)m-1$ directions, which means that the signal $(u_\tau, \ldots, u_{\tau+n})\in \R^{(n+1)m}$ on the left-hand side of \eqref{eq:U_not_spanning} spans persistently only $(n+1)m-1$ directions, namely, $\QQ^{n+1}(\uu)\notin \Dpe{(n+1)m}$.

	\subsection{Proof of Theorem \ref{thm:nec_mi_ct}}\label{proof:nec_mi_ct}
	
	Since the proof is analogous to the one for discrete-time systems, we only highlight where they differ.
	We show that given $\Sigma \in \hctrbx$, if $\DD^n(\uu)\in \C{nm}$ is PPE of degree at most $n^\prime \leq n-1$, then for all $x(0)\in \R^n$, $\xx=\Sigma(\uu, x_0)\notin \Cpe{n}$ regardless of the initial condition. 
	
	\noindent
	I) \emph{$x(t)$ as a linear function of previous inputs.}
	
	Consider $x_0=0$.
	For all $t$, we can write
	\begin{equation}\label{eq:linsys_evolution}
		\begin{split}
			x(t+\delta)  &= e^{A\delta}x(t) + \int_{t}^{t+\delta} e^{A(t+\delta-\tau)}Bu(\tau)\textup{d}\tau.
		\end{split}
	\end{equation}
	By writing the Taylor expansion for $x(t+\delta)$, it holds also
	\begin{equation}\label{eq:taylor_approx}
		\begin{split}
			x(t+\delta) =& x(t) + \dot{x}(t)\delta + \frac{\ddot{x}}{2!}(t)\delta^2 + \ldots\\
			=&e^{A\delta}x(t) + R U(t) + r(t),
		\end{split}
	\end{equation}
	where
	\begin{equation}
		\begin{split}
			R &\coloneqq
			\begin{bmatrix}
				B\delta & \frac{AB}{2}\delta^2 & \ldots & \frac{A^{n-1}B}{n!}\delta^n
			\end{bmatrix},\\
			U(t)&\!\coloneqq \!(u(t), \ldots, \du^{(n-1)}(t)), \; r(t)\!\coloneqq\! \sum_{i=n}^{\infty}o(\delta^i)\du^{(i)}(t).\;\;\;\;\;
		\end{split}
	\end{equation}
	Notice that, given $M\coloneqq \|\dd^n(\xx)\|_\infty$, the remainder $r(t)$ can be bounded by $|r(t)|\leq M\delta^{n}/(n!)$.
	Confronting \eqref{eq:linsys_evolution} and \eqref{eq:taylor_approx}, we obtain
	\begin{equation}\label{eq:integral_approx}
		\int_{t}^{t+\delta} e^{A(t+\delta-\tau)}Bu(\tau)\textup{d}\tau = R U(t) + r(t)
	\end{equation}
	Given any $\bar{T}>0$ (to be chosen later), we can write $x(t+\bar{T})$ as
	\begin{equation}\label{eq:x_as_fcn_of_old_u}
		x(t+\bar{T}) = e^{A\bar{T}}x(t)+\int_t^{t+\bar{T}}e^{A(t+ \bar{T} -\tau)}Bu(\tau)\textup{d}\tau.
	\end{equation}
	Now, let $N \in \N$ (to be chosen later) and $\delta = \bar{T}/N$. Recalling \eqref{eq:integral_approx} to approximate the integral in \eqref{eq:x_as_fcn_of_old_u}, we obtain
	\begin{equation}\label{eq:x_as_fcn_of_U}
		\begin{split}
			x(t&+N\delta) = e^{AN\delta }x(t)+\!\int_{t}^{t+N\delta} e^{A(t+N\delta-\tau)}Bu(\tau)\textup{d}\tau \\
			=&e^{AN\delta }x(t)+
			\!\sum_{i=0}^{N-1}\int_{t+i\delta}^{t+(i+1)\delta} e^{A(t+N\delta-\tau)}Bu(\tau)\textup{d}\tau \\
			=&e^{AN\delta }x(t)+ \sum_{i=0}^{N-1} r(t+i\delta)+\\
			&
			\underbrace{
				\begin{bmatrix}
					e^0 & e^{A\delta} & \ldots & e^{A(N-1)\delta} 
				\end{bmatrix}
			}_{\coloneqq \bar{A}}
			\underbrace{
				(I_N \otimes R)
			}_{\coloneqq \bar{R}}
			\begin{bmatrix}
				U (t+(N-1)\delta) \\
				U (t+(N-2)\delta)\\
				\vdots\\
				U(t)
			\end{bmatrix}.
		\end{split}
	\end{equation}
	
	\noindent
	II) \emph{$x(t)$ can be approximated arbitrarily well by a linear function of the previous inputs.}
	
	We show now the terms $e^{AN\delta }x(t)$ and $\sum_{i=0}^{N-1} r(t+i\delta)$ can be made arbitrarily small. 
	Pick any $\epsilon>0$. Since $e^A$ is Schur and $\xx$ is bounded, there exists a sufficiently large $\bar{T}=N\delta>0$ such that $|e^{A\bar{T}}x(t)|\leq \epsilon$ for all $t \in \R_{\geq 0}$. 
	Next, we choose $N$. It holds
	\begin{equation}
		\left|\sum_{i=0}^{N-1} r(t+i\delta)\right|\leq N \frac{M\delta^n}{n!}=\frac{M\bar{T}^n}{n! N^{n-1}},
	\end{equation} 
	so, given any $\epsilon$ and picking 
	\begin{equation}
		N\geq \sqrt[n-1]{\frac{M\bar{T}^n}{n! \epsilon}},
	\end{equation}
	we obtain $|\sum_{i=0}^{N-1} r(t+i\delta)|\leq \epsilon$.
	
	\noindent
	III) \emph{If the input derivatives are not PPE of degree $n$, $\xx\notin \Cpe{n}$.}
	
	At last, applying Lemma \ref{lemma:not_PE_in_CT}, if $\DD^n(\uu)\notin \Cppe{nm}{n}$, we have that for any $T, \bar{T}, \epsilon>0$ there exists $\bar{N}, \bar{t}>0$ such that, if $N\geq \max(\bar{N}, \sqrt[n-1]{\frac{M\bar{T}^n}{n! \epsilon}})$, we can rewrite \eqref{eq:x_as_fcn_of_U} as
	\begin{equation}
		\begin{split}
			x(\tau+N\delta) =& e^{AN\delta }x(t)+ \sum_{i=0}^{N-1} r(t+i\delta)\\
			& +\bar{A}\bar{R}G\lambda(\tau) + \bar{A}\bar{R}\tilde{X}(\tau),
		\end{split}
	\end{equation}
	where $\lambda(\tau)\in \R^{n-1}$, $|\tilde{X}(\tau)|\leq \epsilon$ for all $\tau \in [\bar{t}, \bar{t}+T]$.
	We can thus conclude that $\xx=\Sigma(\uu, 0) \notin \Cpe{n}$ similarly as done in Theorem \ref{thm:nec_mi_dt} (since $\xx$ is a sum of a signal spanning only $n-1$ directions plus arbitrarily small perturbations).
	To conclude the proof, since $\Sigma$ is a stable linear system, we have that all solutions $\xx=\Sigma(\uu, x(0))$ converge exponentially to those initialized in $x(0)=0$. Since a vanishing term cannot guarantee PE of a signal, we can conclude that for all $x(0)\in \R^n$, if $\DD^n(\uu)\notin \Cppe{nm}{n}$, then $\xx=\Sigma(\uu, x(0)) \notin \Cpe{n}$.
	We omit the second statement since it combines arguments from previous proofs.

	\subsection{Proof of Theorem \ref{thm:suff_mi_ct}}\label{proof:suff_mi_ct}
	
	We prove that $\DD^n(\uu)\in \Cpe{nm} \implies \Sigma(\uu, x(0))\in \Cpe{n}$ for $\Sigma \in \hctrbx$ and for all $x(0)\in \R^n$ by contraposition, i.e., we show that for all $x(0)\in \R^n$, $\Sigma(\uu, x(0))\notin \Cpe{n} \implies \DD^n(\uu)\notin \Cpe{nm}$.
	We prove this in four points.
	
	\noindent
	I) \emph{The lack of PE of $\xx$ constrains the system input.}
	
	If $\xx\coloneqq\Sigma(\uu, x(0))\notin \Cpe{n}$, by applying Lemma~\ref{lemma:PE_of_derivatives} we have that for all $T, \epsilon>0$ we can find a unitary direction $z\in \R^n$ and $\tet>0$ such that
	
	\begin{align}\label{eq:small_x_ct}
		|z^\top \dot{x}(\tau)|\leq \epsilon &&&  |z^\top x(\tau)|\leq \epsilon 
	\end{align}
	
	for all $\tau \in[\tet, \tet+T]$.
	Along direction $z$, the system dynamics read
	\begin{equation}
		z^\top \dot{x}(\tau) = z^\top Ax(\tau) + z^\top Bu(\tau).
	\end{equation}
	If $z^\top B \neq 0$, we can find the input $u(\tau)$ as
	\begin{equation}\label{eq:closed_loop_u_mi_ct}
		u(\tau) = \underbrace{-(z^\top B)^\dagger z^\top A}_{\coloneqq K}x(\tau) + \underbrace{(z^\top B)^\dagger z^\top \dot{x}(\tau)}_{\coloneqq \tu(\tau)} + v(\tau),
	\end{equation}
	where $|\tu(\tau)|\leq |(z^\top B)^\dagger|\epsilon$, and $v(\tau) \in \ker(z^\top B)$.
	
	\noindent
	II) \emph{The closed-loop dynamics depends only on $x(\tau), v(\tau)$.}
	
	The dynamics become
	\begin{equation}\label{eq:closed_loop_dynamics_mi_ct}
		\begin{split}
			u(\tau) &=Kx(\tau) +  \tu(\tau) + v(\tau),\\
			\dot{x}(\tau)&=\underbrace{(A - BK)}_{\coloneqq \tilde{A}}x(\tau) + B\tu(\tau) + Bv(\tau)
		\end{split}
	\end{equation}		
	where, for all $\tau \in[\tet, \tet+T]$, it holds $|B\tu(\tau)|\leq \epsilon$ and $v(\tau) \in \ker(z^\top B)$.
	Consider the signal $\UU(\uu)\coloneqq\DD^n(\uu)$. Given \eqref{eq:closed_loop_u_mi_ct}, we obtain
	\begin{equation}\label{eq:U_expression_ct}
		U(\tau) =
		(I_n \otimes K )
		\begin{bmatrix}
			\dot{x}^{(n-1)}(\tau)\\
			\vdots\\
			x(\tau)
		\end{bmatrix}
		+
		\underbrace{\begin{bmatrix}
				\dot{\tu}^{(n-1)}(\tau)\\
				\vdots\\
				\tu(\tau)
		\end{bmatrix}}_{\coloneqq \tU(\tau)}
		+
		\underbrace{\begin{bmatrix}
				\dot{v}^{(n-1)}(\tau)\\
				\vdots\\
				v(\tau)
		\end{bmatrix}}_{\coloneqq V(\tau)}.
	\end{equation}
	Using \eqref{eq:closed_loop_dynamics_mi_ct}, it holds
	\begin{equation}
		\begin{bmatrix}
			\dot{x}^{(n-1)}(\tau)\\
			\vdots\\
			x(\tau)
		\end{bmatrix}\!=\!
		F\dot{x}^{(n-1)}(\tau)
		+ G\left(
		\tU(\tau) + V(\tau)
		\right),
	\end{equation}
	where 
	\begin{equation}
		F=
		\begin{bmatrix}
			I\\
			\vdots\\
			\tilde{A}^{n-1}
		\end{bmatrix},\;\;\;
		G = 
		\begin{bmatrix}
			0 & \ldots & \ldots & 0\\
			B & 0 & \ldots & 0\\
			\vdots& \vdots & \vdots & \vdots\\
			A^{n-2}B & \ldots & B & 0
		\end{bmatrix},
	\end{equation}
	and substituting in \eqref{eq:U_expression_ct}, we obtain
	\begin{equation}
		\begin{split}
			U(\tau) &\!=\! \underbrace{(I_n \!\otimes\! K)F}_{\coloneqq F^\prime}\dot{x}^{(n-1)}(\tau)\! +\! \underbrace{((I_n \!\otimes \! K)G\!+\!I)}_{G^\prime}
			\left(\!
			\tU(\tau)\!+\!\!
			V(\tau)
			\!\right)\!,\\
			&=
			\begin{bmatrix}
				F^\prime & G^\prime
			\end{bmatrix}
			\begin{bmatrix}
				\dot{x}^{(n-1)}(\tau) \\ 
				V(\tau)
			\end{bmatrix}
			+G^\prime\tU(\tau).
		\end{split}
	\end{equation}
	
	\noindent
	III) \emph{A lack of PE in $\xx$ implies $\DD^{n}(\uu)$ is not PE.}

	Notice that in the period $\tau\in[\tet, \tet+T]$, it holds 
	\begin{enumerate}
		\item[i)] $v(\tau) \in \ker(z^\top B)$, with $\dim(\ker(z^\top B))\leq m-1$, as per \eqref{eq:closed_loop_u_mi_ct}.
		\item[ii)] $|z^\top x(\tau)|\leq \epsilon$ by assumption. 
		\item[iii)] $|\tU(\tau)|\leq \sqrt{n}|(z^\top B)^\dagger|\epsilon$, from \eqref{eq:closed_loop_u_mi_ct}.
	\end{enumerate}
	Holding i) and ii), the space persistently spanned by $(x(\tau), V(\tau))$ is at most $(n-1)+n(m-1)=nm-1$ dimensional. In other words, by following the same procedure as in Lemma \ref{lemma:not_PE_in_CT}, we can write for all $\tau\in[\tet, \tet+T]$
	\begin{equation}
		\begin{bmatrix}
			\dot{x}^{(n-1)}(\tau) \\ 
			V(\tau)
		\end{bmatrix}=
		E \lambda(\tau) + \tilde{\lambda}(\tau),
	\end{equation}
	where $E \in \R^{(nm)\times (nm-1)}$ stacks the directions which are persistently spanned by $(x(\tau), V(\tau))$, $\lambda(\tau) \in \R^{nm-1}$ stacks the projections of $(x(\tau), V(\tau))$ along these directions, and $\tilde{\lambda}(\tau)\in \R^{nm}, |\tilde{\lambda}(\tau)|\leq \epsilon$ is an arbitrarily small perturbation.
	We have
	\begin{equation}
		U(\tau) = 
		\underbrace{\begin{bmatrix}
				F^\prime & G^\prime
			\end{bmatrix} E}_{\coloneqq H}\lambda(\tau) + 
		\begin{bmatrix}
			F^\prime & G^\prime
		\end{bmatrix} \tilde{\lambda}(\tau)
		+G^\prime\tU(\tau),
	\end{equation}
	and since $\begin{bmatrix}
		F^\prime & G^\prime
	\end{bmatrix} E \in \R^{nm\times (nm-1)}$, there exists $z\in \R^{nm}$ such that $z^\top \begin{bmatrix}
		F^\prime & G^\prime
	\end{bmatrix} E=0$, which reads
	\begin{equation}
		z^\top U(\tau) = z^\top
		\begin{bmatrix}
			F^\prime & G^\prime
		\end{bmatrix} \tilde{\lambda}(\tau)
		+z^\top G^\prime\tU(\tau).
	\end{equation}
	Being both $\tU(\tau)$ and $\tilde{\lambda}$ arbitrarily small in the arbitrarily long interval $\tau\in[\tet, \tet+T]$, we conclude $\DD^n(\uu)\notin \Cpe{nm}$.
	
	\noindent
	IV) \emph{The case of $z^\top B=0$.}
	
	We omit the analysis of this case since it follows the same steps done in the fourth point of Proof \ref{proof:suff_mi_dt}.
	
	We move now to the second statement of the theorem. Since it has an analogous  proof, we sketch only the main differences with respect to the proof given before. We want to prove by contraposition that $\DD^n(\uu)\in \Cpe{(n+1)m} \implies \Sigma(\uu, x(0) \in \Cpe{n+m}$ for $\Sigma \in \hctrbxu$ and for all $x(0)\in\R^n$ by showing that for all $x(0)\in \R^n$, $\Sigma(\uu, x(0))\notin \Cpe{n+m}\implies \DD^{n+1}(\uu)\notin \Cpe{(n+1)m}$.
	
	\noindent
	I) \emph{The lack of PE of $(\xx, \uu)$ constrains the system input.}
	
	If $(\xx, \uu)\notin \Cpe{n+m}$ we have that for all $T, \epsilon>0$ we can find $z=(z_x, z_u)\in \R^{n+m}$, $\tet >0$ such that
	\begin{equation}
		z_x^\top x(\tau) = -z_u^\top u(\tau) + \chi(\tau),
	\end{equation}
	where $\chi(\tau) \in \R, |\chi(\tau)|\leq \epsilon$ for all $\tau \in[\tet, \tet+T]$. By isolating $u(\tau)$ and deriving it, we obtain
	\begin{equation}\label{eq:next_u_ct}
		\dot{u}(\tau) = -(z_u^\top)^\dagger z_x^\top \left( Ax(\tau) + Bu(\tau) \right) +(z_u^\top)^\dagger \dot{\chi}(\tau) + \dot{v}(\tau),
	\end{equation}
	where $\dot{v}(\tau)\in \ker(z_u^\top)$. 
	
	\noindent
	II) \emph{The closed-loop dynamics depends only on $x(\tau), v(\tau), u(\tau)$.}
	
	Since we can write $\dot{x}(\tau)=Ax(\tau) + Bu(\tau)$, in the interval $\tau\in[\tet, \tet+T]$ we can use \eqref{eq:next_u_ct} to express each $u(\tau), \dot{u}(\tau), \ldots, \dot{u}^{(n)}(\tau)$ as a linear function of $x(\tau), u(\tau), \dot{v}(\tau), \ldots, \dot{v}^{(n)}(\tau)$ (similarly to \eqref{eq:U_expression_ct}) plus an arbitrarily small quantity $\tU(\tau)$. In other words, we have
	\begin{equation}\label{eq:U_not_spanning_ct}
		\begin{bmatrix}
			u(\tau)\\
			\vdots\\
			\dot{u}^{(n)}(\tau)
		\end{bmatrix}
		=
		K
		\begin{bmatrix}
			u(\tau)\\
			x(\tau)\\
			\dot{v}(\tau)\\
			\vdots\\
			\dot{v}^{(n)}(\tau)
		\end{bmatrix}
		+ \tU(\tau)
	\end{equation}
	
	\noindent
	III) \emph{A lack of PE in $(\xx, \uu)$ implies $\DD^{n+1}(\uu)$ not PE.}
	
	Since $(\xx, \uu)$ is not PE, it spans at most $n+m-1$ directions. Since $v(\tau) \in \ker(z_u^\top)$, the vector $(\dot{v}(\tau), \ldots, \dot{v}^{(n)}(\tau))$ spans at most $(m-1)n$ directions. Overall, we have that the right-hand side of \eqref{eq:U_not_spanning_ct} spans persistently only $n+m-1 + (m-1)n= (n+1)m-1$ directions, which means that the signal $(u(\tau), \ldots, \dot{u}^{(n)}(\tau))\in \R^{(n+1)m}$ on the left-hand side of \eqref{eq:U_not_spanning_ct} spans persistently only $(n+1)m-1$ directions, namely, $\DD^{n+1}(\uu)\notin \Cpe{(n+1)m}$.

	\bibliography{note_on_SR_bib}

\begin{thebibliography}{10}
\providecommand{\url}[1]{#1}
\csname url@samestyle\endcsname
\providecommand{\newblock}{\relax}
\providecommand{\bibinfo}[2]{#2}
\providecommand{\BIBentrySTDinterwordspacing}{\spaceskip=0pt\relax}
\providecommand{\BIBentryALTinterwordstretchfactor}{4}
\providecommand{\BIBentryALTinterwordspacing}{\spaceskip=\fontdimen2\font plus
\BIBentryALTinterwordstretchfactor\fontdimen3\font minus
  \fontdimen4\font\relax}
\providecommand{\BIBforeignlanguage}[2]{{%
\expandafter\ifx\csname l@#1\endcsname\relax
\typeout{** WARNING: IEEEtran.bst: No hyphenation pattern has been}%
\typeout{** loaded for the language `#1'. Using the pattern for}%
\typeout{** the default language instead.}%
\else
\language=\csname l@#1\endcsname
\fi
#2}}
\providecommand{\BIBdecl}{\relax}
\BIBdecl

\bibitem{morgan1977stability}
A.~Morgan and K.~Narendra, ``On the stability of nonautonomous differential
  equations $\dot{\textup{x}}$=[{A}+{B}(t)]x, with skew symmetric matrix
  {B}(t),'' \emph{SIAM Journal on Control and Optimization}, vol.~15, no.~1,
  pp. 163--176, 1977.

\bibitem{morgan1977uniform}
------, ``On the uniform asymptotic stability of certain linear nonautonomous
  differential equations,'' \emph{SIAM Journal on Control and Optimization},
  vol.~15, no.~1, pp. 5--24, 1977.

\bibitem{anderson1977exponential}
B.~Anderson, ``Exponential stability of linear equations arising in adaptive
  identification,'' \emph{IEEE Transactions on Automatic Control}, vol.~22,
  no.~1, pp. 83--88, 1977.

\bibitem{narendra2012stable}
K.~S. Narendra and A.~M. Annaswamy, \emph{Stable adaptive systems}.\hskip 1em
  plus 0.5em minus 0.4em\relax Courier Corporation, 2012.

\bibitem{sastry2011adaptive}
S.~Sastry and M.~Bodson, \emph{Adaptive control: stability, convergence and
  robustness}.\hskip 1em plus 0.5em minus 0.4em\relax Courier Corporation,
  2011.

\bibitem{ioannou1996robust}
P.~A. Ioannou and J.~Sun, \emph{Robust adaptive control}.\hskip 1em plus 0.5em
  minus 0.4em\relax PTR Prentice-Hall Upper Saddle River, NJ, 1996, vol.~1.

\bibitem{bitmead1980lyapunov}
R.~Bitmead and B.~Anderson, ``Lyapunov techniques for the exponential stability
  of linear difference equations with random coefficients,'' \emph{IEEE
  Transactions on Automatic Control}, vol.~25, no.~4, pp. 782--787, 1980.

\bibitem{narendra1987persistent}
K.~S. Narendra and A.~M. Annaswamy, ``Persistent excitation in adaptive
  systems,'' \emph{International Journal of Control}, vol.~45, no.~1, pp.
  127--160, 1987.

\bibitem{loria1999new}
A.~Lor{\'\i}a, E.~Panteley, and A.~Teel, ``A new notion of
  persistency-of-excitation for {UGAS} of {NLTV} systems: Application to
  stabilisation of nonholonomic systems,'' in \emph{1999 European Control
  Conference (ECC)}.\hskip 1em plus 0.5em minus 0.4em\relax IEEE, 1999, pp.
  1363--1368.

\bibitem{loria2000ugas}
------, ``{UGAS} of nonlinear time-varying systems: A $\delta$-persistency of
  excitation approach,'' in \emph{Proceedings of the 39th IEEE Conference on
  Decision and Control}, vol.~4.\hskip 1em plus 0.5em minus 0.4em\relax IEEE,
  2000, pp. 3489--3494.

\bibitem{panteley2001relaxed}
E.~Panteley, A.~Loria, and A.~Teel, ``Relaxed persistency of excitation for
  uniform asymptotic stability,'' \emph{IEEE Transactions on Automatic
  Control}, vol.~46, no.~12, pp. 1874--1886, 2001.

\bibitem{loria2002uniform}
A.~Lor{\i}a and E.~Panteley, ``Uniform exponential stability of linear
  time-varying systems: revisited,'' \emph{Systems \& Control Letters},
  vol.~47, no.~1, pp. 13--24, 2002.

\bibitem{loria2003persistency}
A.~Loria, E.~Panteley, D.~Popovic, and A.~R. Teel, ``Persistency of excitation
  for uniform convergence in nonlinear control systems,'' \emph{arXiv preprint
  math/0301335}, 2003.

\bibitem{loria2005nested}
A.~Lor{\'\i}a, E.~Panteley, D.~Popovic, and A.~R. Teel, ``A nested matrosov
  theorem and persistency of excitation for uniform convergence in stable
  nonautonomous systems,'' \emph{IEEE Transactions on Automatic Control},
  vol.~50, no.~2, pp. 183--198, 2005.

\bibitem{saoud2024hybrid}
A.~Saoud, M.~Maghenem, A.~Loria, and R.~G. Sanfelice, ``Hybrid persistency of
  excitation in adaptive estimation for hybrid systems,'' \emph{IEEE
  Transactions on Automatic Control}, 2024.

\bibitem{aastrom1965numerical}
K.-J. {\AA}str{\"o}m and B.~Torsten, ``Numerical identification of linear
  dynamic systems from normal operating records,'' \emph{IFAC Proceedings
  Volumes}, vol.~2, no.~2, pp. 96--111, 1965.

\bibitem{aastrom1971system}
K.~J. {\AA}str{\"o}m and P.~Eykhoff, ``System identification—a survey,''
  \emph{Automatica}, vol.~7, no.~2, pp. 123--162, 1971.

\bibitem{ljung1990adaptation}
L.~Ljung and S.~Gunnarsson, ``Adaptation and tracking in system
  identification—a survey,'' \emph{Automatica}, vol.~26, no.~1, pp. 7--21,
  1990.

\bibitem{sondhi1976new}
M.~Sondhi and D.~Mitra, ``New results on the performance of a well-known class
  of adaptive filters,'' \emph{Proceedings of the IEEE}, vol.~64, no.~11, pp.
  1583--1597, 1976.

\bibitem{weiss1979digital}
A.~Weiss and D.~Mitra, ``Digital adaptive filters: Conditions for convergence,
  rates of convergence, effects of noise and errors arising from the
  implementation,'' \emph{IEEE Transactions on Information Theory}, vol.~25,
  no.~6, pp. 637--652, 1979.

\bibitem{goodwin1980discrete}
G.~Goodwin, P.~Ramadge, and P.~Caines, ``Discrete-time multivariable adaptive
  control,'' \emph{IEEE Transactions on Automatic Control}, vol.~25, no.~3, pp.
  449--456, 1980.

\bibitem{anderson1982exponential}
B.~D. Anderson and C.~R. Johnson~Jr, ``Exponential convergence of adaptive
  identification and control algorithms,'' \emph{Automatica}, vol.~18, no.~1,
  pp. 1--13, 1982.

\bibitem{elliott1985global}
H.~Elliott, R.~Cristi, and M.~Das, ``Global stability of adaptive pole
  placement algorithms,'' \emph{IEEE Transactions on Automatic Control},
  vol.~30, no.~4, pp. 348--356, 1985.

\bibitem{kreisselmeier1977adaptive}
G.~Kreisselmeier, ``Adaptive observers with exponential rate of convergence,''
  \emph{IEEE Transactions on Automatic Control}, vol.~22, no.~1, pp. 2--8,
  1977.

\bibitem{anderson1977approach}
B.~D. Anderson, ``An approach to multivariable system identification,''
  \emph{Automatica}, vol.~13, no.~4, pp. 401--408, 1977.

\bibitem{boyd1983parameter}
S.~Boyd and S.~Sastry, ``On parameter convergence in adaptive control,''
  \emph{Systems \& Control Letters}, vol.~3, no.~6, pp. 311--319, 1983.

\bibitem{boyd1986necessary}
S.~Boyd and S.~S. Sastry, ``Necessary and sufficient conditions for parameter
  convergence in adaptive control,'' \emph{Automatica}, vol.~22, no.~6, pp.
  629--639, 1986.

\bibitem{kudva1974identification}
P.~Kudva and K.~Narendra, ``An identification procedure for discrete
  multivariable systems,'' \emph{IEEE Transactions on Automatic Control},
  vol.~19, no.~5, pp. 549--552, 1974.

\bibitem{bitmead1984persistence}
R.~Bitmead, ``Persistence of excitation conditions and the convergence of
  adaptive schemes,'' \emph{IEEE Transactions on Information Theory}, vol.~30,
  no.~2, pp. 183--191, 1984.

\bibitem{possieri2022value}
C.~Possieri and M.~Sassano, ``Value iteration for continuous-time linear
  time-invariant systems,'' \emph{IEEE Transactions on Automatic Control},
  vol.~68, no.~5, pp. 3070--3077, 2022.

\bibitem{carnevale2023onpolicy}
L.~Sforni, G.~Carnevale, I.~Notarnicola, and G.~Notarstefano, ``On-policy
  data-driven linear quadratic regulator via combined policy iteration and
  recursive least squares,'' in \emph{2023 62nd IEEE Conference on Decision and
  Control (CDC)}, 2023, pp. 5047--5052.

\bibitem{borghesi2024mr}
M.~Borghesi, A.~Bosso, and G.~Notarstefano, ``{MR}-{ARL}: Model reference
  adaptive reinforcement learning for robustly stable on-policy data-driven
  lqr,'' \emph{arXiv preprint arXiv:2402.14483}, 2024.

\bibitem{willems2005note}
J.~C. Willems, P.~Rapisarda, I.~Markovsky, and B.~L. De~Moor, ``A note on
  persistency of excitation,'' \emph{Systems \& Control Letters}, vol.~54,
  no.~4, pp. 325--329, 2005.

\bibitem{de2019formulas}
C.~De~Persis and P.~Tesi, ``Formulas for data-driven control: Stabilization,
  optimality, and robustness,'' \emph{IEEE Transactions on Automatic Control},
  vol.~65, no.~3, pp. 909--924, 2019.

\bibitem{van2020data}
H.~J. Van~Waarde, J.~Eising, H.~L. Trentelman, and M.~K. Camlibel, ``Data
  informativity: A new perspective on data-driven analysis and control,''
  \emph{IEEE Transactions on Automatic Control}, vol.~65, no.~11, pp.
  4753--4768, 2020.

\bibitem{berberich2020data}
J.~Berberich, J.~K{\"o}hler, M.~A. M{\"u}ller, and F.~Allg{\"o}wer,
  ``Data-driven model predictive control with stability and robustness
  guarantees,'' \emph{IEEE Transactions on Automatic Control}, vol.~66, no.~4,
  pp. 1702--1717, 2020.

\bibitem{van2020willems}
H.~J. Van~Waarde, C.~De~Persis, M.~K. Camlibel, and P.~Tesi, ``Willems’
  fundamental lemma for state-space systems and its extension to multiple
  datasets,'' \emph{IEEE Control Systems Letters}, vol.~4, no.~3, pp. 602--607,
  2020.

\bibitem{de2023learning}
C.~De~Persis and P.~Tesi, ``Learning controllers for nonlinear systems from
  data,'' \emph{Annual Reviews in Control}, p. 100915, 2023.

\bibitem{lopez2023efficient}
V.~G. Lopez, M.~Alsalti, and M.~A. M{\"u}ller, ``Efficient off-policy
  q-learning for data-based discrete-time lqr problems,'' \emph{IEEE
  Transactions on Automatic Control}, vol.~68, no.~5, pp. 2922--2933, 2023.

\bibitem{ljung1971characterization}
L.~Ljung, ``Characterization of the concept of'persistently exciting'in the
  frequency domain,'' 1971.

\bibitem{yuan1977probing}
J.~Yuan and W.~Wonham, ``Probing signals for model reference identification,''
  \emph{IEEE Transactions on Automatic Control}, vol.~22, no.~4, pp. 530--538,
  1977.

\bibitem{nordstrom1987persistency}
N.~Nordstr{\"o}m and S.~S. Sastry, ``Persistency of excitation in possibly
  unstable continuous time systems and parameter convergence in adaptive
  identification,'' in \emph{Adaptive Systems in Control and Signal Processing
  1986}.\hskip 1em plus 0.5em minus 0.4em\relax Elsevier, 1987, pp. 347--352.

\bibitem{moore1983persistence}
J.~Moore, ``Persistence of excitation in extended least squares,'' \emph{IEEE
  Transactions on Automatic Control}, vol.~28, no.~1, pp. 60--68, 1983.

\bibitem{bai1985persistency}
E.-W. Bai and S.~S. Sastry, ``Persistency of excitation, sufficient richness
  and parameter convergence in discrete time adaptive control,'' \emph{Systems
  \& Control Letters}, vol.~6, no.~3, pp. 153--163, 1985.

\bibitem{green1986persistence}
M.~Green and J.~B. Moore, ``Persistence of excitation in linear systems,''
  \emph{Systems \& Control Letters}, vol.~7, no.~5, pp. 351--360, 1986.

\bibitem{mareels1988persistency}
I.~M. Mareels and M.~Gevers, ``Persistency of excitation criteria for linear,
  multivariable, time-varying systems,'' \emph{Mathematics of Control, Signals
  and Systems}, vol.~1, pp. 203--226, 1988.

\bibitem{padoan2017geometric}
A.~Padoan, G.~Scarciotti, and A.~Astolfi, ``A geometric characterization of the
  persistence of excitation condition for the solutions of autonomous
  systems,'' \emph{IEEE Transactions on Automatic Control}, vol.~62, no.~11,
  pp. 5666--5677, 2017.

\bibitem{markovsky2023persistency}
I.~Markovsky, E.~Prieto-Araujo, and F.~D{\"o}rfler, ``On the persistency of
  excitation,'' \emph{Automatica}, vol. 147, p. 110657, 2023.

\bibitem{coulson2022quantitative}
J.~Coulson, H.~J. Van~Waarde, J.~Lygeros, and F.~D{\"o}rfler, ``A quantitative
  notion of persistency of excitation and the robust fundamental lemma,''
  \emph{IEEE Control Systems Letters}, vol.~7, pp. 1243--1248, 2022.

\bibitem{berberich2023quantitative}
J.~Berberich, A.~Iannelli, A.~Padoan, J.~Coulson, F.~D{\"o}rfler, and
  F.~Allg{\"o}wer, ``A quantitative and constructive proof of willems'
  fundamental lemma and its implications,'' in \emph{2023 American Control
  Conference (ACC)}.\hskip 1em plus 0.5em minus 0.4em\relax IEEE, 2023, pp.
  4155--4160.

\bibitem{rapisarda2022persistency}
P.~Rapisarda, M.~K. Camlibel, and H.~Van~Waarde, ``A persistency of excitation
  condition for continuous-time systems,'' \emph{IEEE Control Systems Letters},
  vol.~7, pp. 589--594, 2022.

\bibitem{lopez2022continuous}
V.~G. Lopez and M.~A. M{\"u}ller, ``On a continuous-time version of
  {W}illems’ lemma,'' in \emph{2022 IEEE 61st Conference on Decision and
  Control (CDC)}.\hskip 1em plus 0.5em minus 0.4em\relax IEEE, 2022, pp.
  2759--2764.

\bibitem{sastry1990adaptive}
S.~Sastry and M.~Bodson, ``Adaptive control: stability, convergence, and
  robustness,'' 1990.

\bibitem{Wonham74}
W.~M. Wonham, \emph{Linear multivariable control: a geometric approach}.\hskip
  1em plus 0.5em minus 0.4em\relax Berlin; New York: Springer-Verlag, 1974.

\end{thebibliography}
	\bibliographystyle{IEEEtran}

\vfill\null

\end{document}